\documentclass[12pt,technote, onecolumn, draft]{IEEEtran}
\usepackage{latexsym}
\usepackage{mathrsfs}
\usepackage{multirow}
\usepackage{amsfonts,amssymb,amsmath,amsthm,bm}
\usepackage{color}
\usepackage{cite}
\newtheorem{theorem}{Theorem}
\newtheorem{example}{Example}

\newtheorem{remark}{Remark}
\newtheorem{corollary}{Corollary}
\newtheorem{lemma}{Lemma}
\newtheorem{proposition}{Proposition}

\begin{document}

\title{A Further Study of Vectorial Dual-Bent Functions$^{\dag}$}
\author{Jiaxin Wang, Fang-Wei Fu, Yadi Wei, Jing Yang
\IEEEcompsocitemizethanks{\IEEEcompsocthanksitem Jiaxin Wang, Fang-Wei Fu and Yadi Wei are with Chern Institute of Mathematics and LPMC, Nankai University, Tianjin 300071, China, Emails: wjiaxin@mail.nankai.edu.cn, fwfu@nankai.edu.cn, wydecho@mail.nankai.edu.cn; Jing Yang is with Tsinghua Shenzhen International Graduate School, Tsinghua University, Shenzhen, 518055, China, Email: yangjing@sz.tsinghua.edu.cn
}
\thanks{$^\dag$This research is supported by the National Key Research and Development Program of China (Grant Nos. 2022YFA1005000 and 2018YFA0704703), the National Natural Science Foundation of China (Grant Nos. 12141108, 62371259, 12226336, 62171248), the Fundamental Research Funds for the Central Universities of China (Nankai University), the Nankai Zhide Foundation, Shenzhen Science and Technology Program (JCYJ20220818101012025), and the PCNL KEY project (PCL2021A07).}
\thanks{manuscript submitted  September 23, 2023}
}

\maketitle

\begin{abstract}
  Vectorial dual-bent functions have recently attracted some researchers' interest as they play a significant role in constructing partial difference sets, association schemes, bent partitions and linear codes. In this paper, we further study vectorial dual-bent functions $F: V_{n}^{(p)}\rightarrow V_{m}^{(p)}$, where $2\leq m \leq \frac{n}{2}$, $V_{n}^{(p)}$ denotes an $n$-dimensional vector space over the prime field $\mathbb{F}_{p}$. We give new characterizations of certain vectorial dual-bent functions (called vectorial dual-bent functions with Condition A) in terms of amorphic association schemes, linear codes and generalized Hadamard matrices, respectively. When $p=2$, we characterize vectorial dual-bent functions with Condition A in terms of bent partitions. Furthermore, we characterize certain bent partitions in terms of amorphic association schemes, linear codes and generalized Hadamard matrices, respectively. For general vectorial dual-bent functions $F: V_{n}^{(p)}\rightarrow V_{m}^{(p)}$ with $F(0)=0, F(x)=F(-x)$ and $2\leq m \leq \frac{n}{2}$, we give a necessary and sufficient condition on constructing association schemes. Based on such a result, more association schemes are constructed from vectorial dual-bent functions.
\end{abstract}


\section{Introduction}
\label{sec: 1}
Boolean bent functions were introduced by Rothaus in \cite{Rothaus1976On}, which have been extensively studied due to their important applications in cryptography, coding theory, combinatorics and sequences. Please refer to book \cite{Mesnager2016Be} for further understanding Boolean bent functions and their generalizations, such as $p$-ary bent functions and vectorial bent functions, where $p$ is an odd prime.

As a special class of vectorial bent functions, vectorial dual-bent functions introduced by \c{C}e\c{s}melio\u{g}lu \emph{et al.} \cite{CMP2018Ve} have attracted some researchers' research interest due to their significant applications in constructing partial difference sets \cite{CM2018Be,CMP2021Ve,WF2023Ne}, association schemes \cite{AKMO2023Ve}, bent partitions \cite{WFW2023Be} and linear codes \cite{WSWF2023Co}. Recently, for certain vectorial dual-bent functions $F: V_{n}^{(p)}\rightarrow V_{m}^{(p)}$ (called vectorial dual-bent functions with Condition A), where $V_{n}^{(p)}$ denotes an $n$-dimensional vector space over the prime field $\mathbb{F}_{p}$, Wang \emph{et al.} in  \cite{WFW2023Be} provided a characterization in terms of partial difference sets. Furthermore, when $p$ is an odd prime, they provided a characterization in terms of bent partitions. When $p=2$, they showed that vectorial dual-bent functions with Condition A can be used to construct bent partitions, but they did not give a characterization of vectorial dual-bent functions with Condition A in terms of bent partitions. As far as we know, apart from the literature \cite{WFW2023Be}, there is a lack of research on the characterizations of vectorial dual-bent functions. As to the applications, Anbar \emph{et al.} in \cite{AKMO2023Ve} considered using vectorial dual-bent functions to construct association schemes. Also, they in \cite{AKM2024Am} used bent partitions to construct association schemes. Anbar \emph{et al.} showed that vectorial dual-bent functions $F: V_{n}^{(p)}\rightarrow V_{m}^{(p)}$ with $F(0)=0$ and all component functions $F_{c}, c \in V_{m}^{(p)} \backslash \{0\}$ being regular or weakly regular but not regular (that is, the corresponding $\varepsilon_{F_{c}}, c \in V_{m}^{(p)} \backslash \{0\}$ are all the same) can induce association schemes. Note that for such vectorial dual-bent functions, $n$ must be even. It is interesting to investigate whether there are other vectorial dual-bent functions which can be used to construct association schemes.

In this paper, we further study vectorial dual-bent functions $F: V_{n}^{(p)} \rightarrow V_{m}^{(p)}$, where $2\leq m \leq \frac{n}{2}$. We  summarize our contributions as below.
\begin{itemize}
  \item For any prime $p$, we provide new characterizations of vectorial dual-bent functions $F: V_{n}^{(p)} \rightarrow V_{m}^{(p)}$ with Condition A in terms of amorphic association schemes, linear codes and generalized Hadamard matrices, respectively.
  \item We present the relations between bent partitions of $V_{n}^{(2)}$ of depth $2^m$ and the corresponding vectorial bent functions, based on which we characterize vectorial dual-bent functions with Condition A in terms of bent partitions when $p=2$.
  \item Based on the relations between vectorial dual-bent functions with Condition A and bent partitions, we give new characterizations of certain bent partitions in terms of amorphic association schemes, linear codes and generalized Hadamard matrices, respectively.
  \item For general vectorial dual-bent functions $F: V_{n}^{(p)} \rightarrow V_{m}^{(p)}$ with $F(0)=0$, $F(x)=F(-x)$ and $2\leq m \leq \frac{n}{2}$, a necessary and sufficient condition on constructing association schemes from $F$ is presented. Based on such a result, more association schemes are constructed by using two classes of vectorial dual-bent functions $F: V_{n}^{(p)}\rightarrow V_{m}^{(p)}$ for which $n$ can be odd, or $n$ is even and the corresponding $\varepsilon_{F_{c}}, c \in V_{m}^{(p)} \backslash \{0\}$ are not all the same.
\end{itemize}

The rest of the paper is organized as follows. Section II provides necessary preliminaries. In Sections III-VI, we provide some new characterizations of certain vectorial dual-bent functions. In Section VII, some new characterizations of certain bent partitions are presented. In Section VIII, for certain vectorial dual-bent functions, a necessary and sufficient condition on constructing association schemes is given. In Section IX, we make a conclusion.

\section{Preliminaries}
\label{sec: 2}
In this section, we give the needed results on vectorial dual-bent functions, bent partitions, partial difference sets, association schemes, generalized Hadamard matrices and linear codes, respectively. First, we fix some notations used throughout this paper.
\begin{itemize}
  \item $p$ is a prime and $\zeta_{p}=e^{\frac{2\pi \sqrt{-1}}{p}}$ is a complex primitive $p$-th root of unity.
  \item $\mathbb{F}_{p^n}$ is the finite field with $p^n$ elements.
  \item $\mathbb{F}_{p}^{n}$ is the vector space of the $n$-tuples over $\mathbb{F}_{p}$.
  \item $V_{n}^{(p)}$ is an $n$-dimensional vector space over $\mathbb{F}_{p}$.
  \item $\langle \cdot \rangle_{n}$ denotes a (non-degenerate) inner product of $V_{n}^{(p)}$. In this paper, when $V_{n}^{(p)}=\mathbb{F}_{p^n}$, let $\langle a, b\rangle_{n}=Tr_{1}^{n}(ab)$, where $a, b \in \mathbb{F}_{p^n}$, $Tr_{m}^{n}(\cdot)$ denotes the trace function from $\mathbb{F}_{p^n}$ to $\mathbb{F}_{p^m}$, $m \mid n$; when $V_{n}^{(p)}=\mathbb{F}_{p}^{n}$, let $\langle a, b\rangle_{n}=a \cdot b=\sum_{i=1}^{n}a_{i}b_{i}$, where $a=(a_{1}, \dots, a_{n}), b=(b_{1}, \dots, b_{n})\in \mathbb{F}_{p}^{n}$; when $V_{n}^{(p)}=V_{n_{1}}^{(p)}\times \dots \times V_{n_{s}}^{(p)} (n=\sum_{i=1}^{s}n_{i})$, let $\langle a, b\rangle_{n}=\sum_{i=1}^{s}\langle a_{i}, b_{i}\rangle_{n_{i}}$, where $a=(a_{1}, \dots, a_{s}), b=(b_{1}, \dots, b_{s})\in V_{n}^{(p)}$.
  \item For any set $A\subseteq V_{n}^{(p)}$, let $A^{*}=A \backslash \{0\}$ and $\chi_{u}(A)=\sum_{x \in A}\chi_{u}(x), u \in V_{n}^{(p)}$, where $\chi_{u}$ denotes the character $\chi_{u}(x)=\zeta_{p}^{\langle u, x\rangle_{n}}$.
  \item For a function $F: V_{n}^{(p)}\rightarrow V_{m}^{(p)}$, let $D_{F, i}=\{x \in V_{n}^{(p)}: F(x)=i\}, i \in V_{m}^{(p)}$ and $F({V_{n}^{(p)}}^{*})=\{F(x), x \in {V_{n}^{(p)}}^{*}\}$.
  \item For any set $A$, let $\delta_{A}$ be the indicator function. In particular, if $A=\{a\}$, we simply denote $\delta_{\{a\}}$ by $\delta_{a}$.
\end{itemize}

\subsection{Vectorial dual-bent functions} \label{subsec: 2.1}
A function from $V_{n}^{(p)}$ to $V_{m}^{(p)}$ is called a \textit{vectorial $p$-ary function}, or simply \textit{$p$-ary function} when $m=1$. For a $p$-ary function $f: V_{n}^{(p)}\rightarrow \mathbb{F}_{p}$, the Walsh transform $W_{f}$ is defined as
\begin{equation}\label{1}
  W_{f}(a)=\sum_{x \in V_{n}^{(p)}}\zeta_{p}^{f(x)-\langle a, x\rangle_{n}}, a \in V_{n}^{(p)}.
\end{equation}
The $p$-ary function $f$ can be recovered by the inverse transform
 \begin{equation}\label{2}
  \zeta_{p}^{f(x)}=\frac{1}{p^n}\sum_{a\in V_{n}^{(p)}}W_{f}(a)\zeta_{p}^{\langle a, x\rangle_{n}}, x \in V_{n}^{(p)}.
 \end{equation}

A $p$-ary function $f: V_{n}^{(p)}\rightarrow \mathbb{F}_{p}$ is called \textit{bent} if $|W_{f}(a)|=p^{\frac{n}{2}}$ for any $a \in V_{n}^{(p)}$. When $p=2$, that is, $f$ is a Boolean bent function, then $n$ must be even. The Walsh transform of a $p$-ary bent function $f: V_{n}^{(p)} \rightarrow \mathbb{F}_{p}$ satisfies that when $p=2$, then
\begin{equation}\label{3}
  W_{f}(a)=2^{\frac{n}{2}}(-1)^{f^{*}(a)}, a \in V_{n}^{(2)},
\end{equation}
and when $p$ is an odd prime, then
\begin{equation}\label{4}
  W_{f}(a)=\left\{\begin{split}
                     \pm p^{\frac{n}{2}}\zeta_{p}^{f^{*}(a)}, & \ \text{ if } \ p \equiv 1 \pmod 4 \ \text{or} \ n \ \text{is even},\\
                     \pm \sqrt{-1} p^{\frac{n}{2}} \zeta_{p}^{f^{*}(a)}, & \ \text{ if } \ p \equiv 3 \pmod 4 \ \text{and} \ n \ \text{is odd},
                  \end{split}\right.
\end{equation}
where $f^{*}$ is a $p$-ary function from $V_{n}^{(p)}$ to $\mathbb{F}_{p}$, called the \textit{dual} of $f$. A $p$-ary bent function $f: V_{n}^{(p)}\rightarrow \mathbb{F}_{p}$ is said to be \textit{weakly regular} if $W_{f}(a)=\varepsilon_{f}p^{\frac{n}{2}}\zeta_{p}^{f^{*}(a)}$, where $\varepsilon_{f}$ is a constant independent of $a$, otherwise $f$ is called \textit{non-weakly regular}. In particular, if $W_{f}(a)=p^{\frac{n}{2}}\zeta_{p}^{f^{*}(a)}$, that is, $\varepsilon_{f}=1$, then $f$ is called \textit{regular}. All Boolean bent functions are regular. The dual $f^{*}$ of a weakly regular bent function $f$ is also a weakly regular bent function and
\begin{equation}\label{5}
(f^{*})^{*}(x)=f(-x), \varepsilon_{f^{*}}=\varepsilon_{f}^{-1}.
\end{equation}

A vectorial $p$-ary function $F: V_{n}^{(p)}\rightarrow V_{m}^{(p)}$ is called \textit{vectorial bent} if all \textit{component functions} $F_{c}: V_{n}^{(p)}\rightarrow \mathbb{F}_{p}, c \in {V_{m}^{(p)}}^{*}$ defined as $F_{c}(x)=\langle c, F(x)\rangle_{m}$ are bent. It is known that if $F: V_{n}^{(p)}\rightarrow V_{m}^{(p)}$ is vectorial bent with all component functions $F_{c}, c \in {V_{m}^{(p)}}^{*}$ being regular or weakly regular but not regular (that is, $\varepsilon_{F_{c}}$ is a constant independent of $c$), then $n$ is even and $m\leq \frac{n}{2}$ (see \cite{AKMO2023Ve,CMP2021Ve}). A vectorial $p$-ary bent function $F: V_{n}^{(p)}\rightarrow V_{m}^{(p)}$ is called \textit{vectorial dual-bent} if the set of the duals $(F_{c})^{*}, c \in {V_{m}^{(p)}}^{*}$ of the component functions $F_{c}, c \in {V_{m}^{(p)}}^{*}$ of $F$ (together with the zero function) forms a vector space $\mathcal{V}_{F}$ of bent functions of dimension $m$.

For a vectorial dual-bent function $F: V_{n}^{(p)} \rightarrow V_{m}^{(p)}$, let $\{(F_{c_{1}})^{*}, \dots, (F_{c_{m}})^{*}\}$ be any basis of $\mathcal{V}_{F}$, where $c_{i} \in {V_{m}^{(p)}}^{*}, 1\leq i \leq m$. Then for any $c \in {V_{m}^{(p)}}^{*}$, there is unique nonzero vector $(a_{1}^{(c)}, \dots, a_{m}^{(c)}) \in \mathbb{F}_{p}^{m}$ such that $(F_{c})^{*}=\sum_{i=1}^{m}a_{i}^{(c)}(F_{c_{i}})^{*}$. Define $G: V_{n}^{(p)} \rightarrow V_{m}^{(p)}$ as $G(x)=\sum_{i=1}^{m}(F_{c_{i}})^{*}(x)\alpha_{i}$, where $\{\alpha_{1}, \dots, \alpha_{m}\}$ is any basis of $V_{m}^{(p)}$. For any $c \in {V_{m}^{(p)}}^{*}$, let $\sigma(c) \in {V_{m}^{(p)}}^{*}$ be given by the following equation system:
\begin{equation*}
\left\{
\begin{split}
& \langle \sigma(c), \alpha_{1}\rangle_{m}=a_{1}^{(c)},\\
& \langle \sigma(c), \alpha_{2}\rangle_{m}=a_{2}^{(c)},\\
& \ \ \ \ \ \vdots\\
& \langle \sigma(c), \alpha_{m}\rangle_{m}=a_{m}^{(c)}.\\
\end{split}\right.
\end{equation*}
Then $\sigma$ is a permutation over ${V_{m}^{(p)}}^{*}$ and $(F_{c})^{*}=G_{\sigma(c)}, c \in {V_{m}^{(p)}}^{*}$. Since $F$ is vectorial dual-bent, $(F_{c})^{*}, c \in {V_{m}^{(p)}}^{*}$ are all bent functions and $G$ is vectorial bent. By the argument, one can see that a vectorial bent function $F: V_{n}^{(p)}\rightarrow V_{m}^{(p)}$ is vectorial dual-bent if and only if there exists a vectorial bent function $G: V_{n}^{(p)}\rightarrow V_{m}^{(p)}$ such that $(F_{c})^{*}=G_{\sigma(c)}, c \in {V_{m}^{(p)}}^{*}$, where $\sigma$ is some permutation over ${V_{m}^{(p)}}^{*}$. The vectorial bent function $G$ is called a \textit{vectorial dual} of $F$ and denoted by $F^{*}$. By the above analysis, one can see that the vectorial dual of a vectorial dual-bent function is not unique. In the following, we show that if $F$ is a vectorial dual-bent function for some fixed permutation $\sigma$ over ${V_{m}^{(p)}}^{*}$, then its vectorial dual $F^{*}$ with $(F_{c})^{*}=(F^{*})_{\sigma(c)}, c \in {V_{m}^{(p)}}^{*}$ is unique.

\begin{proposition}\label{1}
Let $F: V_{n}^{(p)}\rightarrow V_{m}^{(p)}$ be a vectorial dual-bent function for some fixed permutation $\sigma$ over ${V_{m}^{(p)}}^{*}$. Then its vectorial dual $F^{*}$ with $(F_{c})^{*}=(F^{*})_{\sigma(c)}, c \in {V_{m}^{(p)}}^{*}$ is unique.
\end{proposition}

\begin{proof}
Let $\{\alpha_{1}, \dots, \alpha_{m}\}$ be any basis of $V_{m}^{(p)}$ and $c_{i}=\sigma^{-1}(\alpha_{i}), 1\leq i \leq m$. Then $\{\sigma(c_{1}), \dots, \sigma(c_{m})\}$ is a basis of $V_{m}^{(p)}$ and for any $x \in V_{n}^{(p)}$, there is unique $G(x) \in V_{m}^{(p)}$ such that the following equation system hold: \\
\begin{equation*}
\left\{
\begin{split}
& \langle G(x), \sigma(c_{1})\rangle_{m}=(F_{c_{1}})^{*}(x),\\
& \langle G(x), \sigma(c_{2})\rangle_{m}=(F_{c_{2}})^{*}(x),\\
& \ \ \ \ \ \vdots\\
& \langle G(x), \sigma(c_{m})\rangle_{m}=(F_{c_{m}})^{*}(x).\\
\end{split}\right.
\end{equation*}
Hence, the vectorial dual $F^{*}$ with $(F_{c})^{*}=(F^{*})_{\sigma(c)}, c \in {V_{m}^{(p)}}^{*}$ is unique and $F^{*}=G$.
\end{proof}

In \cite{WFW2023Be}, Wang \emph{et al.} studied vectorial dual-bent functions for which the corresponding permutation $\sigma$ over ${V_{m}^{(p)}}^{*}$ is the identity map. We recall vectorial dual-bent functions with Condition A defined and studied in \cite{WFW2023Be}.

\textbf{Condition A}: Let $n\geq 4$ be even and $2\leq m \leq \frac{n}{2}$. Let $F: V_{n}^{(p)}\rightarrow V_{m}^{(p)}$ be a vectorial dual-bent function for which
\begin{equation}\label{6}
(F_{c})^{*}=(F^{*})_{c}, c \in {V_{m}^{(p)}}^{*},
\end{equation}
and all component functions $F_{c}, c \in {V_{m}^{(p)}}^{*}$ are regular or weakly regular but not regular. We denote by  $\varepsilon=\varepsilon_{F_{c}}$ for all $c \in {V_{m}^{(p)}}^{*}$.

\begin{remark}\label{1}
Let $n\geq 4$ be even and $2\leq m \leq \frac{n}{2}$. When $p=2$, since all Boolean bent functions are regular, $F: V_{n}^{(2)}\rightarrow V_{m}^{(2)}$ is a vectorial dual-bent function with Condition A if and only if $F$ is a vectorial dual-bent function with $(F_{c})^{*}=(F^{*})_{c}, c \in {V_{m}^{(2)}}^{*}$.
\end{remark}

When $p>3$, if $F: V_{n}^{(p)}\rightarrow V_{m}^{(p)}$ is a vectorial dual-bent function with Condition A, we show that all component functions $F_{c}, c \in {V_{m}^{(p)}}^{*}$ are regular.

\begin{proposition}\label{Proposition 2}
Let $p>3$ be an odd prime. If $F: V_{n}^{(p)}\rightarrow V_{m}^{(p)}$ is a vectorial dual-bent function with Condition A, then all component functions $F_{c}, c \in {V_{m}^{(p)}}^{*}$ are regular, that is, $\varepsilon=1$.
\end{proposition}

\begin{proof}
By the proof of Theorem 1 of \cite{WFW2023Be}, if $F$ is a vectorial dual-bent function with Condition A, then $F(ax)=F(x), a \in \mathbb{F}_{p}^{*}$. Note that $F_{c}(x)-F_{c}(0), c \in {V_{m}^{(p)}}^{*}$ are all weakly regular bent functions with $F_{c}(ax)-F_{c}(0)=F_{c}(x)-F_{c}(0), a \in \mathbb{F}_{p}^{*}$ and $\varepsilon_{F_{c}(x)-F_{c}(0)}=\varepsilon$. By Corollary 3.5 of \cite{HLL2020Ra}, for a weakly regular bent function $f: V_{2r}^{(p)}\rightarrow \mathbb{F}_{p}$ with $f(0)=0, f(ax)=f(x), a \in \mathbb{F}_{p}^{*}$, $f$ is regular if $p>3$. Therefore, we have $\varepsilon=1$ if $p>3$.
\end{proof}

It was shown in \cite{WFW2023Be} that the known bent partitions from (pre)semifields can be obtained from vectorial dual-bent functions with Condition A, and vectorial dual-bent functions with Condition A can be used to construct partial difference sets (see also \cite{WF2023Ne}). In \cite{AKMO2023Ve}, Anbar \emph{et al.} showed that vectorial dual-bent functions with Condition A are able to construct amorphic association schemes. In Sections III-VI, we will further study vectorial dual-bent functions with Condition A.

\subsection{Bent partitions} \label{subsec: 2.2}

Let $n$ be an even positive integer, $K$ be a positive integer divisible by $p$. Let $\Gamma=\{A_{1}, \dots, A_{K}\}$ be a partition of $V_{n}^{(p)}$. Assume that every $p$-ary function $f: V_{n}^{(p)}\rightarrow \mathbb{F}_{p}$ for which every $i\in \mathbb{F}_{p}$ has exactly $\frac{K}{p}$ of sets $A_{j}$ in $\Gamma$ in its preimage set, is a $p$-ary bent function. Then $\Gamma$ is called a \textit{bent partition} of $V_{n}^{(p)}$ of depth $K$ and every such bent function $f$ is called a \textit{bent function constructed from bent partition $\Gamma$}.

For a bent partition $\Gamma=\{A_{i}, 1\leq i \leq p^m\}$ of $V_{n}^{(p)}$, the following lemma gives the cardinality of $A_{i}$.

\begin{lemma}[\cite{AM2022Be}] \label{lemma 1}
Let $n$ be an even positive integer. Let $\Gamma=\{A_{i}, 1\leq i \leq p^m\}$ be a bent partition of $V_{n}^{(p)}$. Then except one set, denoted by $A_{i_{0}}$, all other sets $A_{i}$ have the same cardinality, namely
\begin{equation*}
|A_{i_{0}}|=p^{\frac{n}{2}-m}(p^{\frac{n}{2}}\mp 1)\pm p^{\frac{n}{2}}, \ |A_{i}|=p^{\frac{n}{2}-m}(p^{\frac{n}{2}}\mp 1), i \neq i_{0}.
\end{equation*}
\end{lemma}

In \cite{WFW2023Be}, Wang \emph{et al.} studied the relations between vectorial dual-bent functions with Condition A and bent partitions with Condition $\mathcal{C}$. We recall bent partitions with Condition $\mathcal{C}$ defined and studied in \cite{WFW2023Be}.

\textbf{Condition $\mathcal{C}$}: Let $n\geq 4$ be even, $2\leq m \leq \frac{n}{2}$. Let $\Gamma=\{A_{i}, i \in V_{m}^{(p)}\}$ be a bent partition of $V_{n}^{(p)}$, which satisfies that $aA_{i}=A_{i}$ for any $a \in \mathbb{F}_{p}^{*}$ and $i \in V_{m}^{(p)}$ and all bent functions constructed from $\Gamma$ are regular or weakly regular but not regular. We denote by $\varepsilon=\varepsilon_{f}$ for all bent functions $f$ constructed from $\Gamma$.

\begin{remark}\label{2}
Let $n\geq 4$ be even and $2\leq m \leq \frac{n}{2}$. When $p=2$, since all Boolean bent functions are regular, Condition $\mathcal{C}$ is trivial for every bent partition of $V_{n}^{(2)}$ of depth $2^m$.
\end{remark}

When $p$ is odd, it was proved in \cite{WFW2023Be} that bent partitions with Condition $\mathcal{C}$ one-to-one correspond to vectorial dual-bent functions with Condition A.

\begin{lemma}[\cite{WFW2023Be}]\label{Lemma 2}
Let $p$ be an odd prime. Let $\Gamma=\{A_{i}, i \in V_{m}^{(p)}\}$ be a partition of $V_{n}^{(p)}$, where $n \geq 4$ is even and $2\leq m \leq \frac{n}{2}$. Define $F: V_{n}^{(p)}\rightarrow V_{m}^{(p)}$ as $F(x)=\sum_{i \in V_{m}^{(p)}}\delta_{A_{i}}(x)i$. Then $\Gamma$ is a bent partition with Condition $\mathcal{C}$ if and only if $F$ is a vectorial dual-bent function with Condition A.
\end{lemma}

\subsection{Partial difference sets and association schemes} \label{subsec: 2.3}

Let $(G, +)$ be a finite abelian group of order $v$ and $D$ be a subset of $G$ with $k$ elements. Then $D$ is called a \textit{$(v,k,\lambda,\mu)$ partial difference set} of $G$, if $\overline{D} \ \overline{(-D)}=\mu \overline{G}+(\lambda-\mu)\overline{D}+\gamma 0$ with $\overline{D}=\sum_{g \in D}g$ denoting the element in the group ring $\mathbb{Z}[G]$ and $-D=\{-d, d \in D\}$, where $\gamma=k-\mu$ if $0 \notin D$ and $\gamma=k-\lambda$ if $0 \in D$. By Page 223 of \cite{Ma1994A}, the empty set can be seen as a $(v, 0, \lambda, 0)$ partial difference set of any finite abelian group of order $v$, where $\lambda$ is any integer. A partial difference set $D$ is called \textit{regular} if $-D=D$ and $0 \notin D$. A regular $(v, k, \lambda, \mu)$ partial difference set is called to be of \textit{Latin square type} if $v=N^{2}, k=s(N-1), \lambda=N+s^{2}-3s, \mu=s^{2}-s$, and a regular $(v, k, \lambda, \mu)$ partial difference set is called to be of \textit{negative Latin square type} if $v=N^{2}, k=s(N+1), \lambda=-N+s^{2}+3s, \mu=s^{2}+s$. We allow $s=0$, which corresponds to the empty set.

There is an important tool to characterize partial difference sets in terms of characters.

\begin{lemma}[\cite{Ma1994A,Tan2010St}]\label{Lemma 3}
Let $G$ be an abelian group of order $v$. Suppose that $D$ is a subset of $G$ with $k$ elements which satisfies $-D=D$ and $0 \notin D$. Then $D$ is a $(v, k, \lambda, \mu)$ partial difference set if and only if for each non-principal character $\chi$ of $G$,
\begin{equation*}
  \chi(D)=\frac{\beta\pm \sqrt{\Delta}}{2},
\end{equation*}
where $\chi(D)=\sum_{x \in D}\chi(x)$, $\beta=\lambda-\mu, \gamma=k-\mu, \Delta=\beta^{2}+4\gamma$.
\end{lemma}

Let $X$ be a nonempty finite set. A \textit{$d$-class association scheme} on $X$ is a sequence $R_{0}, R_{1}, \dots, R_{d}$ of nonempty subsets
of $X \times X$, satisfying

1. $R_{0}=\{(x,x): x \in X\}$;

2. $X\times X=R_{0} \bigcup R_{1} \bigcup \dots \bigcup R_{d}$ and $R_{i} \bigcap R_{j}=\emptyset$ for $i\neq j$;

3. for any $i \in \{0, \dots, d\}$, there is $j$ such that $R_{i}^{\top}=R_{j}$, where $R_{i}^{\top}=\{(y, x): (x, y) \in R_{i}\}$;

4. for all integers $k, i, j \in \{0, 1, \dots, d\}$, and for all $x, y \in X$ such that $(x, y) \in R_{k}$,
the number $p_{i, j}^{k}=|\{z \in X: (x, z) \in R_{i}, (z,y) \in R_{j}\}|$ depends only on $k, i, j$ and not on $(x, y)$.

The numbers $p_{i, j}^{k}$ are called \textit{intersection numbers} of an association scheme. If for any $i \in \{0, \dots, d\}$, $R_{i}^{\top}=R_{i}$, then the association scheme is called \textit{symmetric}.

A \textit{fusion} of an association scheme $\{R_{0}, R_{1}, \dots, R_{d}\}$ on $X$ is a partition $\{A_{0}, A_{1}, \dots, A_{t}\}$ of $X \times X$ such that $A_{0}=R_{0}$ and each $A_{i}$ ($1\leq i \leq t$) is the union of some of $R_{j}, 1\leq j \leq d$. An association scheme is called \textit{amorphic} if its any fusion is again an association scheme. The following lemma gives a characterization of amorphic association schemes induced from partitions.

\begin{lemma}[\cite{VDM2010So,VanDam2003St}]\label{Lemma 4}
Let nonempty sets $D_{0}=\{0\}, D_{1}, \dots, D_{d}$ form a partition of a finite abelian group $G$, where $d\geq 3$. Define $R_{i}, 0\leq i \leq d$ as
\begin{equation*}
R_{i}=\{(x, y) \in G \times G: x-y \in D_{i}\}.
\end{equation*}
The following two statements are equivalent.

(1) $R_{0}, R_{1}, \dots, R_{d}$ form an amorphic association scheme.

(2) $D_{1}, \dots, D_{d}$ are regular partial difference sets, all of which are of Latin square type, or all of which are of negative Latin square type.
\end{lemma}

\subsection{Generalized Hadamard matrices} \label{subsec: 2.4}

Let $\zeta_{m}=e^{\frac{2\pi \sqrt{-1}}{m}}$ be a complex primitive $m$-th root of unity. A complex matrix $H$ of size $n \times n$ consisting of integer powers of $\zeta_{m}$ is called a \textit{generalized Hadamard matrix} if $H\overline{H}^{\top}=nI_{n}$, where $\overline{H}$ is the conjugate matrix of $H$, $\overline{H}^{\top}$ is the transpose matrix of $\overline{H}$, and $I_{n}$ is the identity matrix of size $n \times n$. When $m=2$, $H$ is simply called a \textit{Hadamard matrix}.

There is a characterization of $p$-ary bent functions in terms of generalized Hadamard matrices.

\begin{lemma}[\cite{Dillon1974El,Natalia2015Al}] \label{Lemma 5}
Let $f: V_{n}^{(p)}\rightarrow \mathbb{F}_{p}$. Define $H=\left[\zeta_{p}^{f(x-y)}\right]_{x, y \in V_{n}^{(p)}}$. Then $f$ is a $p$-ary bent function if and only if $H$ is a generalized Hadamard matrix.
\end{lemma}

\subsection{Linear codes} \label{subsec: 2.5}

For a vector $a=(a_{1}, \dots, a_{n})\in \mathbb{F}_{p}^{n}$, the \textit{Hamming weight} of $a$ is defined as $wt(a)=|\{1\leq i\leq n: a_{i}\neq 0\}|$. For two vectors $a, b \in \mathbb{F}_{p}^{n}$, the \textit{Hamming distance} between $a$ and $b$ is defined as $d(a, b)=wt(a-b)$.

Let $C$ be a $p$-ary $[n, k]$ linear code, that is, $C$ is a subspace of $\mathbb{F}_{p}^{n}$ with dimension $k$. The \textit{minimum Hamming distance} $d$ of $C$ is defined as $d=\min\{d(a, b): a, b \in C, a \neq b\}=\min\{wt(c): c\in C, c \neq 0\}$. The dual code of $C$ is defined by $C^{\perp}=\{u\in \mathbb{F}_{p}^{n}:u \cdot c=0\ \text{for all} \ c\in C\}$. If the minimum Hamming weight $d^{\bot}$ of the dual code $C^{\perp}$ satisfies $d^{\bot}\geq 3$, then $C$ is called \textit{projective}. For any $1\leq i \leq n$, let $A_{i}$ denote the number of codewords in $C$ whose Hamming weight is $i$. The sequence $(1, A_{1}, \dots, A_{n})$ is called the \textit{weight distribution} of $C$. The code $C$ is called \textit{$t$-weight} if $|\{1\leq i\leq n: A_{i}\neq 0\}|=t$.

The following lemma gives a characterization of a two-weight projective $p$-ary linear code in terms of a partial difference set.

\begin{lemma}[\cite{Ma1994A}]\label{Lemma 6}
Let $\widetilde{D}=\{d_{1}, \dots, d_{m}\}$, where $d_{i}, 1\leq i \leq m$ are pairwise linearly independent vectors in $V_{n}^{(p)}$. Define
\begin{equation*}
C_{\widetilde{D}}=\{(\langle x, d_{1}\rangle_{n}, \dots, \langle x, d_{m}\rangle_{n}): x \in V_{n}^{(p)}\}.
\end{equation*}
Then $C_{\widetilde{D}}$ is a two-weight $[m, n]$ projective linear code if and only if $D=\mathbb{F}_{p}^{*}\widetilde{D}=\{yd_{i}: y \in \mathbb{F}_{p}^{*}, 1\leq i \leq m\}$ is a regular partial difference set in $V_{n}^{(p)}$. Furthermore, if the two nonzero weights of $C_{\widetilde{D}}$ are $w_{1}$ and $w_{2}$, then the parameters of the $(v, k, \lambda, \mu)$ partial difference set $D$ are $v=p^{n}, k=m(p-1), \lambda=k^{2}+3k-p(k+1)(w_{1}+w_{2})+p^{2}w_{1}w_{2}, \mu=k^{2}+k-pk(w_{1}+w_{2})+p^{2}w_{1}w_{2}$.
\end{lemma}

\section{A Characterization of vectorial dual-bent functions with Condition A in terms of amorphic association schemes} \label{sec: 3}

In this section, we give a characterization of vectorial dual-bent functions with Condition A in terms of amorphic association schemes.

In Theorem 6 of \cite{WFW2023Be}, Wang \emph{et al.} characterized vectorial dual-bent functions $F$ with Condition A in terms of partial difference sets $D_{F, I}^{*}$, where $I$ is an arbitrary nonempty subset of $V_{m}^{(p)}$ and $D_{F,I}=\bigcup_{i \in I}D_{F, i}$. In the following, based on Theorem 6 of \cite{WFW2023Be}, we give a characterization of vectorial dual-bent functions with Condition A in terms of partial difference sets $D_{F, i}^{*}, i \in V_{m}^{(p)}$.

\begin{proposition}\label{Proposition 3}
Let $F: V_{n}^{(p)}\rightarrow V_{m}^{(p)}$, where $n\geq 4$ is even and $2\leq m\leq \frac{n}{2}$. The following two statements are equivalent.

\emph{(1)} $F$ is a vectorial dual-bent function with Condition A.

\emph{(2)} For any $i \in V_{m}^{(p)}$, $D_{F, i}^{*}$ is a regular $(p^{n}, s_{i}(p^{\frac{n}{2}}-\varepsilon), \varepsilon p^{\frac{n}{2}}+s_{i}^{2}-3\varepsilon s_{i}, s_{i}^{2}-\varepsilon s_{i})$ partial difference set, where $s_{i}=p^{\frac{n}{2}-m}+\varepsilon \delta_{F(0)}(i)$, $\varepsilon \in \{\pm 1\}$ is a constant with $\varepsilon=1$ if $p\neq 3$.
\end{proposition}

\begin{proof}
If (1) holds, then (2) holds by Theorem 6 of \cite{WFW2023Be} and Proposition 2. In the following, we prove that if (2) holds, then (1) holds.

By Lemma 3, if (2) holds, then for any $u \in V_{n}^{(p)}$ and $i \in V_{m}^{(p)}$, we have $\chi_{u}(D_{F, i})=p^{n-m}\delta_{0}(u)+\varepsilon p^{\frac{n}{2}}-\varepsilon p^{\frac{n}{2}-m}$ or $\chi_{u}(D_{F, i})=p^{n-m}\delta_{0}(u)-\varepsilon p^{\frac{n}{2}-m}$.
For any $u \in V_{n}^{(p)}$, let
\begin{equation*}
n_{u}=|\{i \in V_{m}^{(p)}: \chi_{u}(D_{F, i})=p^{n-m}\delta_{0}(u)+\varepsilon p^{\frac{n}{2}}-\varepsilon p^{\frac{n}{2}-m}\}|.
\end{equation*}
Since $p^{n}\delta_{0}(u)=\sum_{x \in V_{n}^{(p)}}\zeta_{p}^{\langle u, x\rangle_{n}}=\chi_{u}(V_{n}^{(p)})$ and
\begin{equation*}
\begin{split}
\chi_{u}(V_{n}^{(p)})&=\sum_{i \in V_{m}^{(p)}}\chi_{u}(D_{F, i})\\
&=n_{u}(p^{n-m}\delta_{0}(u)+\varepsilon p^{\frac{n}{2}}-\varepsilon p^{\frac{n}{2}-m})+(p^{m}-n_{u})(p^{n-m}\delta_{0}(u)-\varepsilon p^{\frac{n}{2}-m})\\
&=p^{n}\delta_{0}(u)+\varepsilon p^{\frac{n}{2}}(n_{u}-1),
\end{split}
\end{equation*}
we have $n_{u}=1$. Therefore, for any nonempty set $I\subseteq V_{m}^{(p)}$ and $u \in V_{n}^{(p)}$,
\begin{equation}\label{7}
\chi_{u}(D_{F, I})=p^{n-m}\delta_{0}(u)|I|+\varepsilon p^{\frac{n}{2}}-\varepsilon p^{\frac{n}{2}-m}|I| \text{ or }
\chi_{u}(D_{F, I})=p^{n-m}\delta_{0}(u)|I|-\varepsilon p^{\frac{n}{2}-m}|I|,
\end{equation}
where $D_{F, I}=\sum_{i \in I}D_{F, i}$. For any $i \in V_{m}^{(p)}$, define $E_{i}=\{u \in V_{n}^{(p)}: \chi_{u}(D_{F, i})=p^{n-m}\delta_{0}(u)+\varepsilon p^{\frac{n}{2}}-\varepsilon p^{\frac{n}{2}-m}\}$. We claim that $E_{i}\bigcap E_{j}=\emptyset$ for any $i \neq j$ and $\bigcup_{i \in V_{m}^{(p)}}E_{i}=V_{n}^{(p)}$. If there exists $i\neq j$ such that $E_{i}\bigcap E_{j}\neq \emptyset$, then there is $u \in V_{n}^{(p)}$ such that $\chi_{u}(D_{F, i})=\chi_{u}(D_{F, j})=p^{n-m}\delta_{0}(u)+\varepsilon p^{\frac{n}{2}}-\varepsilon p^{\frac{n}{2}-m}$ and $\chi_{u}(D_{F, i}\bigcup D_{F, j})=2p^{n-m}\delta_{0}(u)+2\varepsilon p^{\frac{n}{2}}-2\varepsilon p^{\frac{n}{2}-m}$, which contradicts Eq. (7). Thus, $E_{i}\bigcap E_{j}=\emptyset$ for any $i\neq j$. If there is $u \in V_{n}^{(p)}$ such that $u \notin E_{i}$ for any $i \in V_{m}^{(p)}$, then $\chi_{u}(D_{F, i})=p^{n-m}\delta_{0}(u)-\varepsilon p^{\frac{n}{2}-m}$ for any $i \in V_{m}^{(p)}$ and $\chi_{u}(V_{n}^{(p)})=\sum_{i \in V_{m}^{(p)}}\chi_{u}(D_{F, i})=p^{n}\delta_{0}(u)-\varepsilon p^{\frac{n}{2}}$, which contradicts $\chi_{u}(V_{n}^{(p)})=\sum_{x \in V_{n}^{(p)}}\zeta_{p}^{\langle u, x\rangle_{n}}=p^{n}\delta_{0}(u)$. Thus, $\bigcup_{i \in V_{m}^{(p)}}E_{i}=V_{n}^{(p)}$. By the above arguments, we can obtain
\begin{equation*}
  \chi_{u}(D_{F, I})=p^{n-m}\delta_{0}(u)|I|+\varepsilon p^{\frac{n}{2}-m}(p^m\delta_{E_{I}}(u)-|I|),
\end{equation*}
where $E_{I}=\sum_{i \in I}E_{i}$. Then by Lemma 1 of \cite{WFW2023Be}, $F$ is a vectorial dual-bent function with Condition A.
\end{proof}

\begin{remark}\label{Remark 3}
For a vectorial dual-bent function $F: V_{n}^{(p)}\rightarrow V_{m}^{(p)}$ with Condition A, by Proposition 3, $D_{F, i}^{*}, i \in V_{m}^{(p)}$ are all regular partial difference sets of Latin square type if $\varepsilon=1$, and $D_{F, i}^{*}, i \in V_{m}^{(p)}$ are all regular partial difference sets of negative Latin square type if $\varepsilon=-1$.
\end{remark}

The following corollary is directly from Proposition 3 and Remark 1.

\begin{corollary}\label{Corollary 1}
Let $F: V_{n}^{(2)}\rightarrow V_{m}^{(2)}$, where $n\geq 4$ is even and $2\leq m\leq \frac{n}{2}$. The following two statements are equivalent.

\emph{(1)} $F$ is a vectorial dual-bent function with $(F_{c})^{*}=(F^{*})_{c}, c \in {V_{m}^{(2)}}^{*}$.

\emph{(2)} For any $i \in V_{m}^{(2)}$, $D_{F, i}^{*}$ is a regular $(2^{n}, s_{i}(2^{\frac{n}{2}}-1), 2^{\frac{n}{2}}+s_{i}^{2}-3s_{i}, s_{i}^{2}-s_{i})$ partial difference set, where $s_{i}=2^{\frac{n}{2}-m}+\delta_{F(0)}(i)$.
\end{corollary}

Based on Proposition 3 and Lemma 4, we give the following theorem, which characterizes vectorial dual-bent functions with Condition A in terms of amorphic association schemes.

\begin{theorem}\label{Theorem 1}
Let $F: V_{n}^{(p)}\rightarrow V_{m}^{(p)}$, where $n\geq 4$ is even and $2\leq m \leq \frac{n}{2}$. Denote $I=F({V_{n}^{(p)}}^{*})$. Define
\begin{equation*}
\begin{split}
& R_{id}=\{(x, x): x \in V_{n}^{(p)}\}, \\
& R_{i}=\{(x, y): x, y \in V_{n}^{(p)}, x-y \in D_{F, i}^{*}\}, i \in I.
\end{split}
\end{equation*}
The following two statements are equivalent.

\emph{(1)} $F$ is a vectorial dual-bent function with Condition A.

\emph{(2)} $\{R_{id}, R_{i}, i \in I\}$ is an $|I|$-class amorphic association scheme for which $|I|\geq 3$ and for any $i \in I$, the intersection number $p_{i, i}^{id}=p^{n-m}-\varepsilon p^{\frac{n}{2}-m}+\delta_{F(0)}(i)(\varepsilon p^{\frac{n}{2}}-1)$, where $\varepsilon \in \{\pm 1\}$ is a constant with $\varepsilon=1$ if $p\neq 3$.

Furthermore, if \emph{(1)} or \emph{(2)} holds, then the following statement holds:

\emph{(3)} $I=V_{m}^{(p)}$ and $|I|=p^m$ except one case that $p=3$, $n=2m$ and $\varepsilon=-1$ (in such a case, $I=V_{m}^{(3)}\backslash \{F(0)\}$ and $|I|=3^m-1$).
\end{theorem}

\begin{proof}
If (1) holds, then (2) and (3) follow from Theorem 3 of \cite{AKMO2023Ve}, Proposition 4 of \cite{WFW2023Be} and Proposition 2. In the following, we prove that if (2) holds, then (1) and (3) hold.

Note that $i \in I$ if and only if $D_{F, i}^{*}$ is nonempty. By Lemma 4, $\{R_{id}, R_{i}, i \in I\}$ is an amorphic association scheme if and only if for any $i \in I$, $D_{F, i}^{*}$ is a regular $(p^n, s_{i}(p^{\frac{n}{2}}-\varepsilon'), \varepsilon' p^{\frac{n}{2}}+s_{i}^{2}-3\varepsilon' s_{i}, s_{i}^{2}-\varepsilon' s_{i})$ partial difference set, where $s_{i}$ is a positive integer, $\varepsilon' \in \{\pm 1\}$ is a constant. From $|D_{F, i}^{*}|=p_{i, i}^{id}=p^{n-m}-\varepsilon p^{\frac{n}{2}-m}+\delta_{F(0)}(i)(\varepsilon p^{\frac{n}{2}}-1), i \in I$, we have
\begin{equation*}
s_{i}(p^{\frac{n}{2}}-\varepsilon')=p^{n-m}-\varepsilon p^{\frac{n}{2}-m}+\delta_{F(0)}(i)(\varepsilon p^{\frac{n}{2}}-1), i \in I.
\end{equation*}
Assume that $\varepsilon'\neq \varepsilon$, that is, $\varepsilon'=-\varepsilon$. Let $i \in I \backslash\{F(0)\}$, then $s_{i}(p^{\frac{n}{2}}+\varepsilon)=p^{\frac{n}{2}-m}(p^{\frac{n}{2}}-\varepsilon)$. Since $gcd(p^{\frac{n}{2}}+\varepsilon, p^{\frac{n}{2}-m})=1$, then $(p^{\frac{n}{2}}+\varepsilon) \mid (p^{\frac{n}{2}}-\varepsilon)$, which implies that $\varepsilon=-1$ and $p=3, n=2$, which contradicts $n\geq 4$. Therefore, $\varepsilon'=\varepsilon$ and for any $i \in I$, $D_{F, i}^{*}$ is a regular $(p^n, s_{i}(p^{\frac{n}{2}}-\varepsilon), \varepsilon p^{\frac{n}{2}}+s_{i}^{2}-3\varepsilon s_{i}, s_{i}^{2}-\varepsilon s_{i})$ partial difference set, where $s_{i}=p^{\frac{n}{2}-m}+\varepsilon \delta_{F(0)}(i)$. Note that $0 \in D_{F, F(0)}$ and $|D_{F, i}^{*}|=0$ if $i \in V_{m}^{(p)}\backslash I$.

If $D_{F, F(0)}=\{0\}$, then $F(0) \notin I$ and
\begin{equation}\label{8}
\begin{split}
p^n=|V_{n}^{(p)}|&=\sum_{i \in V_{m}^{(p)}}|D_{F, i}|\\
&=\sum_{i \in V_{m}^{(p)} \backslash \{F(0)\}}|D_{F, i}^{*}|+|D_{F, F(0)}|\\
&=\sum_{i \in I}|D_{F, i}^{*}|+1\\
&=p^{\frac{n}{2}-m}(p^{\frac{n}{2}}-\varepsilon)|I|+1.
\end{split}
\end{equation}
From Eq. (8), we have $p^{\frac{n}{2}-m} \mid (p^n-1)$, which implies that $n=2m$. Further, by Eq. (8) and $|I|\leq p^m-1=p^\frac{n}{2}-1$, we obtain $\varepsilon=-1, |I|=p^m-1, I=V_{m}^{(p)} \backslash \{F(0)\}$. Note that $p=3$ since $\varepsilon=1$ when $p\neq 3$. In this case, for any $i \in V_{m}^{(3)}$, $D_{F, i}^{*}$ is a regular $(3^n, s_{i}(3^{\frac{n}{2}}+1), -3^{\frac{n}{2}}+s_{i}^{2}+3s_{i}, s_{i}^{2}+s_{i})$ partial difference set, where $s_{i}=1-\delta_{F(0)}(i)$. By Proposition 3, $F$ is a vectorial dual-bent function with Condition A.

If $D_{F, F(0)}^{*}$ is nonempty, then $F(0) \in I$ and
\begin{equation*}
\begin{split}
p^n=|V_{n}^{(p)}|&=\sum_{i \in V_{m}^{(p)}}|D_{F, i}|\\
&=\sum_{i \in V_{m}^{(p)}\backslash \{F(0)\}}|D_{F, i}^{*}|+|D_{F, F(0)}|\\
&=\sum_{i \in I \backslash \{F(0)\}}|D_{F, i}^{*}|+|D_{F, F(0)}|\\
&=p^{\frac{n}{2}-m}(p^{\frac{n}{2}}-\varepsilon)(|I|-1)+(p^{\frac{n}{2}-m}+\varepsilon)(p^{\frac{n}{2}}-\varepsilon)+1,
\end{split}
\end{equation*}
which implies that $|I|=p^m$, that is, $I=V_{m}^{(p)}$. Thus for any $i \in V_{m}^{(p)}$, $D_{F, i}^{*}$ is a regular $(p^n, s_{i}(p^{\frac{n}{2}}-\varepsilon), \varepsilon p^{\frac{n}{2}}+s_{i}^{2}-3\varepsilon s_{i}, s_{i}^{2}-\varepsilon s_{i})$ partial difference set, where $s_{i}=p^{\frac{n}{2}-m}+\varepsilon \delta_{F(0)}(i)$. By Proposition 3, $F$ is a vectorial dual-bent function with Condition A.

Furthermore, the analysis mentioned above shows that statement (3) holds.
\end{proof}

\begin{remark}\label{Remark 4}
Keep the same notation as in Theorem 1. If $F$ is a vectorial dual-bent function with Condition A, then by Corollary 1 of \emph{\cite{VDM2010So}}, the intersection numbers of the amorphic association scheme induced from $F$ are given by the following equations:
\begin{equation*}
\begin{split}
& p_{id, id}^{i}=p_{id, i}^{j}=p_{i, id}^{j}=p_{i, j}^{id}=0, p_{id, i}^{i}=p_{i, id}^{i}=p_{id, id}^{id}=1, p_{i, i}^{id}=s_{i}(p^{\frac{n}{2}}-\varepsilon), \\
& p_{i, i}^{i}=\varepsilon p^{\frac{n}{2}}-2+(s_{i}-\varepsilon)(s_{i}-2\varepsilon), p_{i, i}^{j}=s_{i}(s_{i}-\varepsilon), p_{i, j}^{j}=s_{i}(s_{j}-\varepsilon), p_{i, j}^{k}=s_{i}s_{j},
\end{split}
\end{equation*}
where $i, j, k \in F({V_{n}^{(p)}}^{*})$ are distinct, $s_{i}=p^{\frac{n}{2}-m}+\varepsilon \delta_{F(0)}(i)$, $\varepsilon \in \{\pm 1\}$ is a constant with $\varepsilon=1$ if $p\neq 3$.
\end{remark}

The following corollary is directly from Theorem 1 and Remark 1.

\begin{corollary}\label{Corollary 2}
Let $F: V_{n}^{(2)}\rightarrow V_{m}^{(2)}$, where $n\geq 4$ is even and $2\leq m\leq \frac{n}{2}$.
Define
\begin{equation*}
\begin{split}
& R_{id}=\{(x, x): x \in V_{n}^{(2)}\}, \\
& R_{i}=\{(x, y): x, y \in V_{n}^{(2)}, x+y \in D_{F, i}^{*}\}, i \in V_{m}^{(2)}.
\end{split}
\end{equation*}
The following two statements are equivalent.

\emph{(1)} $F$ is a vectorial dual-bent function with $(F_{c})^{*}=(F^{*})_{c}, c \in {V_{m}^{(2)}}^{*}$.

\emph{(2)} $\{R_{id}, R_{i}, i \in V_{m}^{(2)}\}$ is a $2^m$-class amorphic association scheme for which for any $i \in V_{m}^{(2)}$, the intersection number $p_{i, i}^{id}=2^{n-m}-2^{\frac{n}{2}-m}+\delta_{F(0)}(i)(2^{\frac{n}{2}}-1)$.
\end{corollary}

\section{A characterization of vectorial dual-bent functions with Condition A in terms of linear codes} \label{sec: 4}
In this section, we give a characterization of vectorial dual-bent functions with Condition A in terms of linear codes.

First, we introduce a notation. For a set $D\subseteq {V_{n}^{(p)}}^{*}$, let $\widetilde{D}$ be a subset of $D$, denoted by $\widetilde{D}=\{x_{1}, \dots, x_{t}\}$, for which $x_{j}, 1\leq j \leq t$ are pairwise linearly independent, and for any $x \in D$, there exist $a \in \mathbb{F}_{p}^{*}$ and $x_{j}$ such that $x=ax_{j}$. Note that when $p=2$, $\widetilde{D}=D$.

\begin{theorem}\label{Theorem 2}
Let $F: V_{n}^{(p)}\rightarrow V_{m}^{(p)}$, where $n\geq 4$ is even and $2\leq m \leq \frac{n}{2}$. Denote $I=F({V_{n}^{(p)}}^{*})$. Define
\begin{equation*}
C_{\widetilde{D_{F, i}^{*}}}=\{c_{\alpha}=(\langle \alpha, x\rangle_{n})_{x \in \widetilde{D_{F, i}^{*}}}: \alpha \in V_{n}^{(p)}\}, i \in I.
\end{equation*}
The following two statements are equivalent.

\emph{(1)} $F$ is a vectorial dual-bent function with Condition A.

\emph{(2)} For any $i \in I$, $C_{\widetilde{D_{F, i}^{*}}}$ is a two-weight $[\frac{p^{n-m}-\varepsilon p^{\frac{n}{2}-m}+\delta_{F(0)}(i)(\varepsilon p^{\frac{n}{2}}-1)}{p-1}, n]$ projective linear code and the two nonzero weights are
\begin{equation*}
\begin{split}
& w_{1}=p^{n-m-1}+\frac{1-\varepsilon +2\varepsilon \delta_{F(0)}(i)}{2}p^{\frac{n}{2}-1},\\
& w_{2}=p^{n-m-1}+\frac{-1-\varepsilon +2\varepsilon \delta_{F(0)}(i)}{2}p^{\frac{n}{2}-1},
\end{split}
\end{equation*}
where $\varepsilon \in \{\pm 1\}$ is a constant with $\varepsilon=1$ if $p\neq 3$.

Furthermore, if \emph{(1)} or \emph{(2)} holds, then the following statement holds:

\emph{(3)} $I=V_{m}^{(p)}$ except one case that $p=3$, $n=2m$ and $\varepsilon=-1$ (in such a case, $I=V_{m}^{(3)}\backslash \{F(0)\}$).
\end{theorem}

\begin{proof}
(1) $\Rightarrow$ (2): If (1) holds, then (3) holds by Proposition 4 of \cite{WFW2023Be} and Proposition 2. Since $F$ is a vectorial dual-bent function with Condition A, by the proof of Theorem 1 of \cite{WFW2023Be}, we have $F(ax)=F(x), a \in \mathbb{F}_{p}^{*}$, which implies that $D_{F, i}^{*}=\mathbb{F}_{p}^{*}\widetilde{D_{F, i}^{*}}, i \in I$. For any $i \in I$, by Proposition 3, we have that $D_{F, i}^{*}$ is a regular $(p^n, s_{i}(p^{\frac{n}{2}}-\varepsilon), \varepsilon p^{\frac{n}{2}}+s_{i}^{2}-3\varepsilon s_{i}, s_{i}^{2}-\varepsilon s_{i})$ partial difference set, where $s_{i}=p^{\frac{n}{2}-m}+\varepsilon \delta_{F(0)}(i)$, $\varepsilon \in \{\pm 1\}$ is a constant with $\varepsilon=1$ if $p\neq 3$. By Lemma 6, $C_{\widetilde{D_{F, i}^{*}}}$ is a two-weight projective linear code with parameters $[\frac{p^{n-m}-\varepsilon p^{\frac{n}{2}-m}+\delta_{F(0)}(i)(\varepsilon p^{\frac{n}{2}}-1)}{p-1}, n]$ and the two nonzero weights are $w_{1}=p^{n-m-1}+\frac{1-\varepsilon +2\varepsilon \delta_{F(0)}(i)}{2}p^{\frac{n}{2}-1}$, $w_{2}=p^{n-m-1}+\frac{-1-\varepsilon +2\varepsilon \delta_{F(0)}(i)}{2}p^{\frac{n}{2}-1}$.

(2) $\Rightarrow$ (1): If (2) holds, then by Lemma 6, for any $i \in I$, $\mathbb{F}_{p}^{*}\widetilde{D_{F, i}^{*}}$ is a regular $(p^n, s_{i}(p^{\frac{n}{2}}-\varepsilon), \varepsilon p^{\frac{n}{2}}+s_{i}^{2}-3\varepsilon s_{i}, s_{i}^{2}-\varepsilon s_{i})$ partial difference set, where $s_{i}=p^{\frac{n}{2}-m}+\varepsilon \delta_{F(0)}(i)$. Note that for any $i \in I$, $D_{F, i}^{*}\subseteq \mathbb{F}_{p}^{*}\widetilde{D_{F, i}^{*}}$ and $|D_{F, i}^{*}|\leq |\mathbb{F}_{p}^{*}\widetilde{D_{F, i}^{*}}|=s_{i}(p^{\frac{n}{2}}-\varepsilon)$. Then
\begin{equation*}
\begin{split}
p^n-1=&\sum_{i \in I}|D_{F, i}^{*}|\leq \sum_{i \in I}|\mathbb{F}_{p}^{*}\widetilde{D_{F, i}^{*}}|=\sum_{i \in I}s_{i}(p^{\frac{n}{2}}-\varepsilon)=(p^{\frac{n}{2}}-\varepsilon)\sum_{i \in I}p^{\frac{n}{2}-m}+\varepsilon \delta_{F(0)}(i),
\end{split}
\end{equation*}
which yields that
\begin{equation}\label{9}
|I|p^{\frac{n}{2}-m}+\varepsilon \delta_{I}(F(0))\geq p^{\frac{n}{2}}+\varepsilon.
\end{equation}

If $F(0) \notin I$, then $D_{F, F(0)}=\{0\}$ and by Inequality (9), $|I|p^{\frac{n}{2}-m}\geq p^{\frac{n}{2}}+\varepsilon$. By $|I|\leq p^m-1$, we obtain $\varepsilon=-1, n=2m, |I|=p^m-1, I=V_{m}^{(p)} \backslash\{F(0)\}$. Note that $p=3$ since $\varepsilon=1$ when $p\neq 3$. In this case, $\sum_{i \in V_{m}^{(3)} \backslash\{F(0)\}}|D_{F, i}^{*}|=\sum_{i \in V_{m}^{(3)} \backslash\{F(0)\}}|\mathbb{F}_{3}^{*}\widetilde{D_{F, i}^{*}}|$, which implies that for any $i \in V_{m}^{(3)} \backslash \{F(0)\}$, we have $D_{F, i}^{*}=\mathbb{F}_{3}^{*}\widetilde{D_{F, i}^{*}}$. Therefore, for any $i \in V_{m}^{(3)}$, $D_{F,i}^{*}$ is a regular $(3^n, s_{i}(3^{\frac{n}{2}}+1), -3^{\frac{n}{2}}+s_{i}^{2}+3s_{i}, s_{i}^{2}+s_{i})$ partial difference set, where $s_{i}=1-\delta_{F(0)}(i)$. By Proposition 3, $F$ is a vectorial dual-bent function with Condition A.

If $F(0) \in I$, then $D_{F, F(0)}^{*}$ is nonempty and by Inequality (9), $|I|p^{\frac{n}{2}-m}+\varepsilon \geq p^{\frac{n}{2}}+\varepsilon$, that is, $|I|\geq p^m$. By $|I|\leq p^m$, we have $|I|=p^m$, $I=V_{m}^{(p)}$. In this case, $\sum_{i \in V_{m}^{(p)}}|D_{F, i}^{*}|=\sum_{i \in V_{m}^{(p)}}|\mathbb{F}_{p}^{*}\widetilde{D_{F, i}^{*}}|$, which implies that for any $i \in V_{m}^{(p)}$, we have $D_{F, i}^{*}=\mathbb{F}_{p}^{*}\widetilde{D_{F, i}^{*}}$. Therefore, for any $i \in V_{m}^{(p)}$, $D_{F,i}^{*}$ is a regular $(p^n, s_{i}(p^{\frac{n}{2}}-\varepsilon), \varepsilon p^{\frac{n}{2}}+s_{i}^{2}-3\varepsilon s_{i}, s_{i}^{2}-\varepsilon s_{i})$ partial difference set, where $s_{i}=p^{\frac{n}{2}-m}+\varepsilon \delta_{F(0)}(i)$. By Proposition 3, $F$ is a vectorial dual-bent function with Condition A.

Furthermore, the analysis mentioned above shows that statement (3) holds.
\end{proof}

The following corollary is directly from Theorem 2 and Remark 1.

\begin{corollary}\label{Corollary 3}
Let $F: V_{n}^{(2)}\rightarrow V_{m}^{(2)}$, where $n\geq 4$ is even and $2\leq m \leq \frac{n}{2}$. Define
\begin{equation*}
C_{D_{F, i}^{*}}=\{c_{\alpha}=(\langle \alpha, x\rangle_{n})_{x \in D_{F, i}^{*}}: \alpha \in V_{n}^{(2)}\}, i \in V_{m}^{(2)}.
\end{equation*}
The following two statements are equivalent.

\emph{(1)} $F$ is a vectorial dual-bent function with $(F_{c})^{*}=(F^{*})_{c}, c \in {V_{m}^{(2)}}^{*}$.

\emph{(2)} For any $i \in V_{m}^{(2)}$, $C_{D_{F, i}^{*}}$ is a two-weight $[2^{n-m}-2^{\frac{n}{2}-m}+\delta_{F(0)}(i)(2^{\frac{n}{2}}-1), n]$ projective linear code and the two nonzero weights are $w_{1}=2^{n-m-1}, w_{2}=2^{n-m-1}-2^{\frac{n}{2}-1}+\delta_{F(0)}(i)2^{\frac{n}{2}}$.
\end{corollary}

\section{A characterization of vectorial dual-bent functions with Condition A in terms of generalized Hadamrad matrices} \label{sec: 5}

In this section, we give a characterization of vectorial dual-bent functions with Condition A in terms of generalized Hadamard matrices.

Since the case of $p$ being odd is more complicated, we first consider the case of $p=2$.

\begin{theorem}\label{Theorem 3}
Let $F: V_{n}^{(2)}\rightarrow V_{m}^{(2)}$, where $n\geq 4$ is even and $2\leq m \leq \frac{n}{2}$. For any $c \in {V_{m}^{(2)}}^{*}$, define
\begin{equation*}
H_{c}=\left[(-1)^{F_{c}(x+y)}\right]_{x, y \in V_{n}^{(2)}},
\end{equation*}
where $F_{c}(x)=\langle c, F(x)\rangle_{m}$. The following two statements are equivalent.

\emph{(1)} $F$ is a vectorial dual-bent function with $(F_{c})^{*}=(F^{*})_{c}, c \in {V_{m}^{(2)}}^{*}$.

\emph{(2)} $H_{c}, c \in {V_{m}^{(2)}}^{*}$ are all Hadamard matrices and for any $c\neq d \in {V_{m}^{(2)}}^{*}$, $H_{c}H_{d}=2^{\frac{n}{2}}H_{c+d}$.
\end{theorem}

\begin{proof}
(1) $\Rightarrow$ (2): Since $F$ is a vectorial bent function, that is, $F_{c}, c \in {V_{m}^{(2)}}^{*}$ are all Boolean bent functions, by Lemma 5 we have that $H_{c}, c \in {V_{m}^{(2)}}^{*}$ are all Hadamard matrices. For a matrix $M=[a_{i, j}]$, denote $a_{i, j}$ by $(M)_{i, j}$. For any $c\neq d \in {V_{m}^{(2)}}^{*}$ and $i, j \in V_{n}^{(2)}$, since $F$ is vectorial bent, we have
\begin{equation}\label{10}
\begin{split}
(H_{c}H_{d})_{i,j}&=\sum_{u \in V_{n}^{(2)}}(-1)^{F_{c}(u+i)+F_{d}(u+j)}\\
&=2^{-2n}\sum_{u \in V_{n}^{(2)}}\sum_{x \in V_{n}^{(2)}}W_{F_{c}}(x)(-1)^{\langle u+i, x\rangle_{n}}\sum_{y \in V_{n}^{(2)}}W_{F_{d}}(y)(-1)^{\langle u+j, y\rangle_{n}}\\
&=2^{-n}\sum_{u \in V_{n}^{(2)}}\sum_{x, y \in V_{n}^{(2)}}(-1)^{(F_{c})^{*}(x)+(F_{d})^{*}(y)+\langle u+i, x\rangle_{n}+\langle u+j, y\rangle_{n}}\\
&=2^{-n}\sum_{x, y \in V_{n}^{(2)}}(-1)^{(F_{c})^{*}(x)+(F_{d})^{*}(y)+\langle i, x\rangle_{n}+\langle j, y\rangle_{n}}\sum_{u \in V_{n}^{(2)}}(-1)^{\langle u, x+y\rangle_{n}}\\
&=\sum_{x \in V_{n}^{(2)}}(-1)^{(F_{c})^{*}(x)+(F_{d})^{*}(x)+\langle i+j, x\rangle_{n}}\\
&=W_{(F_{c})^{*}+(F_{d})^{*}}(i+j).
\end{split}
\end{equation}
Since $F$ is vectorial dual-bent with $(F_{c})^{*}=(F^{*})_{c}, c \in {V_{m}^{(2)}}^{*}$, we have $(F_{c})^{*}+(F_{d})^{*}=(F^{*})_{c}+(F^{*})_{d}=(F^{*})_{c+d}=(F_{c+d})^{*}$. Thus by Eq. (5) and Eq. (10) we obtain
\begin{equation*}
(H_{c}H_{d})_{i, j}=W_{(F_{c+d})^{*}}(i+j)=2^{\frac{n}{2}}(-1)^{F_{c+d}(i+j)}=2^{\frac{n}{2}}(H_{c+d})_{i, j},
\end{equation*}
which implies that $H_{c}H_{d}=2^{\frac{n}{2}}H_{c+d}$.

(2) $\Rightarrow$ (1): Since $H_{c}, c \in {V_{m}^{(2)}}^{*}$ are all Hadamard matrices, by Lemma 5 we have that $F_{c}, c \in {V_{m}^{(2)}}^{*}$ are all Boolean bent functions, that is, $F$ is vectorial bent. For any $c\neq d \in {V_{m}^{(2)}}^{*}$ and $i, j \in V_{n}^{(2)}$, from Eq. (10) and $H_{c}H_{d}=2^{\frac{n}{2}}H_{c+d}$, we have
\begin{equation}\label{11}
W_{(F_{c})^{*}+(F_{d})^{*}}(i+j)=2^{\frac{n}{2}}(-1)^{F_{c+d}(i+j)}.
\end{equation}
By Eq. (11), for any $c \neq d \in {V_{m}^{(2)}}^{*}$, we have that
$(F_{c})^{*}+(F_{d})^{*}$ is a Boolean bent function and $((F_{c})^{*}+(F_{d})^{*})^{*}=F_{c+d}$, which implies that
\begin{equation}\label{12}
(F_{c})^{*}+(F_{d})^{*}=(F_{c+d})^{*}.
\end{equation}
Let $\{\alpha_{1}, \dots, \alpha_{m}\}$ be an arbitrary fixed basis of $V_{m}^{(2)}$. For any $x \in V_{n}^{(2)}$, let $G(x) \in V_{m}^{(2)}$ be given by the following equation system \\
\begin{equation*}
\left\{
\begin{split}
& \langle \alpha_{1}, G(x)\rangle_{m}=(F_{\alpha_{1}})^{*}(x),\\
& \langle \alpha_{2}, G(x)\rangle_{m}=(F_{\alpha_{2}})^{*}(x),\\
& \ \ \ \ \ \vdots\\
& \langle \alpha_{m}, G(x)\rangle_{m}=(F_{\alpha_{m}})^{*}(x).
\end{split}\right.
\end{equation*}
Then $G$ is a function from $V_{n}^{(2)}$ to $V_{m}^{(2)}$ satisfying $G_{\alpha_{i}}=(F_{\alpha_{i}})^{*}, 1\leq i \leq m$. For any $c \in {V_{m}^{(2)}}^{*}$, denote $c$ by $c=\alpha_{i_{1}}+\dots+\alpha_{i_{t}}$. Then
\begin{equation}\label{13}
\begin{split}
G_{c}(x)&=\langle c, G(x)\rangle_{m}\\
&=\langle \alpha_{i_{1}}, G(x)\rangle_{m}+\dots+\langle \alpha_{i_{t}}, G(x)\rangle_{m}\\
&=G_{\alpha_{i_{1}}}(x)+\dots+G_{\alpha_{i_{t}}}(x)\\
&=(F_{\alpha_{i_{1}}})^{*}(x)+\dots+(F_{\alpha_{i_{t}}})^{*}(x).
\end{split}
\end{equation}
Combine Eq. (12) and Eq. (13), we obtain
\begin{equation*}
G_{c}(x)=(F_{\alpha_{i_{1}}+\dots+\alpha_{i_{t}}})^{*}(x)=(F_{c})^{*}(x).
\end{equation*}
For any $c \in {V_{m}^{(2)}}^{*}$, since $G_{c}=(F_{c})^{*}$ is a Boolean bent function, we have that $G$ is vectorial bent. Therefore, $F$ is  vectorial dual-bent with $(F_{c})^{*}=(F^{*})_{c}, c \in {V_{m}^{(2)}}^{*}$, where $F^{*}=G$.
\end{proof}

Below we give an example to illustrate Theorem 3.

\begin{example}\label{Example 1}
Let $F: \mathbb{F}_{2^6} \times \mathbb{F}_{2^6}\rightarrow \mathbb{F}_{2^2}$ be defined by
\begin{equation*}
F(x_{1}, x_{2})=Tr_{2}^{6}(x_{1}x_{2}^{58}).
\end{equation*}
Then by Proposition 3 of \emph{\cite{WFW2023Be}}, $F$ is a vectorial dual-bent function with Condition A. For any $c \in \mathbb{F}_{2^2}^{*}$, define
\begin{equation*}
H_{c}=\left[(-1)^{Tr_{1}^{6}(c(x_{1}+y_{1})(x_{2}+y_{2})^{58})}\right]_{(x_{1}, x_{2}), (y_{1}, y_{2}) \in \mathbb{F}_{2^6} \times \mathbb{F}_{2^6}}.
\end{equation*}
Then by Theorem 3, $H_{c}, c \in \mathbb{F}_{2^2}^{*}$ are all Hadamard matrices and $H_{c}H_{d}=64H_{c+d}$ for any $c\neq d \in \mathbb{F}_{2^2}^{*}$.
\end{example}

In the following, we consider the case of $p$ being odd. First, we need a lemma.

For an odd prime $p$, define $U_{p}^{(1)}=\{\zeta_{p}^{i}: 0\leq i \leq p-1\}$ and $U_{p}^{(-1)}=\{-\zeta_{p}^{i}: 0\leq i \leq p-1\}$.

\begin{lemma}\label{Lemma 7}
Let $p$ be an odd prime. Let $F: V_{n}^{(p)}\rightarrow V_{m}^{(p)}$, where $n\geq 4$ is even and $2\leq m \leq \frac{n}{2}$. For any $c \in {V_{m}^{(p)}}^{*}$ and $z \in V_{n}^{(p)}$, define
\begin{equation} \label{14}
H_{c}^{(z)}=\left[\zeta_{p}^{F_{c}(x-y)-\langle z, x-y\rangle_{n}}\right]_{x, y \in V_{n}^{(p)}},
\end{equation}
where $F_{c}(x)=\langle c, F(x)\rangle_{m}$.
The following two statements are equivalent.

\emph{(1)} $F$ is a vectorial bent function with all component functions $F_{c}, c \in {V_{m}^{(p)}}^{*}$ being regular or weakly regular but not regular, that is, $\varepsilon_{F_{c}}=\varepsilon$ for all $c \in {V_{m}^{(p)}}^{*}$, where $\varepsilon \in \{\pm 1\}$ is a constant.

\emph{(2)} $H_{c}^{(z)}, c \in {V_{m}^{(p)}}^{*}, z \in V_{n}^{(p)}$ are all generalized Hadamard matrices for which there exists a constant $\varepsilon \in \{\pm 1\}$ such that
\begin{equation}\label{15}
p^{-\frac{n}{2}}\sum_{i \in V_{n}^{(p)}}(H_{c}^{(z)})_{i, 0} \in U_{p}^{(\varepsilon)} \ \text{ for all } \ c \in {V_{m}^{(p)}}^{*}, z \in V_{n}^{(p)},
\end{equation}
where for a matrix $M=[a_{i, j}]$, denote $a_{i, j}$ by $(M)_{i, j}$.
\end{lemma}

\begin{proof}
(1) $\Rightarrow$ (2): Since $F$ is a vectorial bent function, $F_{c}, c \in {V_{m}^{(p)}}^{*}$ are all bent functions. Further, $F_{c}(x)-\langle z, x\rangle_{n}$ is bent for any $c \in {V_{m}^{(p)}}^{*}$ and $z \in V_{n}^{(p)}$ since $W_{F_{c}(x)-\langle z, x\rangle_{n}}(a)=W_{F_{c}}(z+a), a \in V_{n}^{(p)}$. By Lemma 5, we have that $H_{c}^{(z)}, c \in {V_{m}^{(p)}}^{*}, z \in V_{n}^{(p)}$ are all generalized Hadamard matrices. For any $c \in {V_{m}^{(p)}}^{*}$ and $z \in V_{n}^{(p)}$, we have
\begin{equation}\label{16}
p^{-\frac{n}{2}}\sum_{i \in V_{n}^{(p)}}(H_{c}^{(z)})_{i, 0}=p^{-\frac{n}{2}}\sum_{i \in V_{n}^{(p)}}\zeta_{p}^{F_{c}(i)- \langle z, i\rangle_{n}}
=p^{-\frac{n}{2}}W_{F_{c}}(z).
\end{equation}
Since $F$ is vectorial bent with $\varepsilon_{F_{c}}=\varepsilon, c \in {V_{m}^{(p)}}^{*}$, by Eq. (16) we have that Eq. (15) holds.

(2) $\Rightarrow$ (1): Since $H_{c}^{(0)}, c \in {V_{m}^{(p)}}^{*}$ are all generalized Hadamard matrices, by Lemma 5 we have that $F_{c}, c \in {V_{m}^{(p)}}^{*}$ are all bent functions, then $F$ is vectorial bent. By Eq. (15) and Eq. (16), we have that for any $c \in {V_{m}^{(p)}}^{*}$, the component function $F_{c}$ is weakly regular with $\varepsilon_{F_{c}}=\varepsilon$.
\end{proof}

Based on Lemma 7, with a similar proof as Theorem 3, we give the following theorem, which provides a characterization of vectorial dual-bent functions with Condition A in terms of generalized Hadamard matrices when $p$ is an odd prime.

\begin{theorem}\label{Theorem 4}
Let $p$ be an odd prime. Let $F: V_{n}^{(p)}\rightarrow V_{m}^{(p)}$, where $n\geq 4$ is even and $2\leq m \leq \frac{n}{2}$. Let matrices $H_{c}^{(z)}, c \in {V_{m}^{(p)}}^{*}, z \in V_{n}^{(p)}$ be defined by \emph{Eq. (14)}. For simplicity, denote $H_{c}^{(0)}$ by $H_{c}$. The following two statements are equivalent.

\emph{(1)} $F$ is a vectorial dual-bent function with Condition A.

\emph{(2)} $H_{c}^{(z)}, c \in {V_{m}^{(p)}}^{*}, z \in V_{n}^{(p)}$ are all generalized Hadamard matrices for which there exists a constant $\varepsilon \in \{\pm 1\}$ with $\varepsilon=1$ if $p\neq 3$ such that \emph{Eq. (15)} holds and $H_{c}\overline{H_{d}}^{\top}=\varepsilon p^{\frac{n}{2}}H_{c-d}$ for any $c\neq d \in {V_{m}^{(p)}}^{*}$.
\end{theorem}

\begin{proof}
(1) $\Rightarrow$ (2): Since $F$ is a vectorial dual-bent function with Condition A, by Proposition 2, we have that for any $c \in {V_{m}^{(p)}}^{*}$, $F_{c}$ is weakly regular bent with $\varepsilon_{F_{c}}=\varepsilon$, where $\varepsilon \in \{\pm 1\}$ is a constant with $\varepsilon=1$ if $p\neq 3$. By Lemma 7, we have that $H_{c}^{(z)}, c \in {V_{m}^{(p)}}^{*}, z \in V_{n}^{(p)}$ are all generalized Hadamard matrices and Eq. (15) holds. For any $c\neq d \in {V_{m}^{(p)}}^{*}$ and $i, j \in V_{n}^{(p)}$, since $F$ is vectorial bent with $\varepsilon_{F_{c}}=\varepsilon, c \in {V_{m}^{(p)}}^{*}$, we have
\begin{equation}\label{17}
\begin{split}
(H_{c}\overline{H_{d}}^{\top})_{i,j}&=\sum_{u \in V_{n}^{(p)}}\zeta_{p}^{F_{c}(i-u)-F_{d}(j-u)}\\
&=p^{-2n}\sum_{u \in V_{n}^{(p)}}\sum_{x \in V_{n}^{(p)}}W_{F_{c}}(x)\zeta_{p}^{\langle i-u, x\rangle_{n}}\sum_{y \in V_{n}^{(p)}}\overline{W_{F_{d}}(y)}\zeta_{p}^{-\langle j-u, y\rangle_{n}}\\
&=p^{-n}\sum_{u \in V_{n}^{(p)}}\sum_{x\in V_{n}^{(p)}}\varepsilon \zeta_{p}^{(F_{c})^{*}(x)+\langle i-u, x\rangle_{n}}
\sum_{y \in V_{n}^{(p)}}\varepsilon \zeta_{p}^{-(F_{d})^{*}(y)-\langle j-u, y\rangle_{n}}\\
&=p^{-n}\sum_{x, y \in V_{n}^{(p)}}\zeta_{p}^{(F_{c})^{*}(x)-(F_{d})^{*}(y)+\langle i, x\rangle_{n}-\langle j, y\rangle_{n}}\sum_{u \in V_{n}^{(p)}}\zeta_{p}^{\langle u, y-x\rangle_{n}}\\
&=\sum_{x \in V_{n}^{(p)}}\zeta_{p}^{(F_{c})^{*}(x)-(F_{d})^{*}(x)+\langle i-j, x\rangle_{n}}\\
&=W_{(F_{c})^{*}-(F_{d})^{*}}(j-i).
\end{split}
\end{equation}
Since $F$ is vectorial dual-bent with $(F_{c})^{*}=(F^{*})_{c}, c \in {V_{m}^{(p)}}^{*}$, for any $c\neq d \in {V_{m}^{(p)}}^{*}$ we have $(F_{c})^{*}-(F_{d})^{*}=(F^{*})_{c}-(F^{*})_{d}=(F^{*})_{c-d}=(F_{c-d})^{*}$. Hence by Eq. (5) and Eq. (17) we have
\begin{equation*}
(H_{c}\overline{H_{d}}^{\top})_{i, j}=W_{(F_{c-d})^{*}}(j-i)=\varepsilon p^{\frac{n}{2}}\zeta_{p}^{F_{c-d}(i-j)}=\varepsilon p^{\frac{n}{2}}(H_{c-d})_{i, j},
\end{equation*}
which implies that $H_{c}\overline{H_{d}}^{\top}=\varepsilon p^{\frac{n}{2}}H_{c-d}$.

(2) $\Rightarrow$ (1): Since $H_{c}^{(z)}, c \in {V_{m}^{(p)}}^{*}, z \in V_{n}^{(p)}$ are all generalized Hadamard matrices and there is a constant $\varepsilon \in \{\pm 1\}$ such that Eq. (15) holds, by Lemma 7 we have that $F$ is a vectorial bent function for which for any $c \in {V_{m}^{(p)}}^{*}$, $F_{c}$ is weakly regular with $\varepsilon_{F_{c}}=\varepsilon$. For any $c\neq d \in {V_{m}^{(p)}}^{*}$ and $i, j \in V_{n}^{(p)}$, by Eq. (17) and $H_{c}\overline{H_{d}}^{\top}=\varepsilon p^{\frac{n}{2}}H_{c-d}$, we have
\begin{equation}\label{18}
W_{(F_{c})^{*}-(F_{d})^{*}}(j-i)=\varepsilon p^{\frac{n}{2}}\zeta_{p}^{F_{c-d}(i-j)}.
\end{equation}
By Eq. (18), for any $c \neq d \in {V_{m}^{(p)}}^{*}$, we have that
$(F_{c})^{*}-(F_{d})^{*}$ is a weakly regular bent function with $\varepsilon_{(F_{c})^{*}-(F_{d})^{*}}=\varepsilon$ and $((F_{c})^{*}-(F_{d})^{*})^{*}(x)=F_{c-d}(-x)$, which implies that
\begin{equation}\label{19}
(F_{c})^{*}-(F_{d})^{*}=(F_{c-d})^{*}.
\end{equation}
Since Eq. (19) holds for all $c\neq d \in {V_{m}^{(p)}}^{*}$, for any $c, d \in {V_{m}^{(p)}}^{*}$ with $c+d \neq 0$, we have
\begin{equation}\label{20}
(F_{c})^{*}+(F_{d})^{*}=(F_{c+d})^{*}.
\end{equation}
Let $\{\alpha_{1}, \dots, \alpha_{m}\}$ be an arbitrary fixed basis of $V_{m}^{(p)}$. For any $x \in V_{n}^{(p)}$, let $G(x) \in V_{m}^{(p)}$ be given by the following equation system
\begin{equation*}
\left\{
\begin{split}
& \langle \alpha_{1}, G(x)\rangle_{m}=(F_{\alpha_{1}})^{*}(x),\\
& \langle \alpha_{2}, G(x)\rangle_{m}=(F_{\alpha_{2}})^{*}(x),\\
& \ \ \ \ \ \vdots\\
& \langle \alpha_{m}, G(x)\rangle_{m}=(F_{\alpha_{m}})^{*}(x).
\end{split}\right.
\end{equation*}
Then $G$ is a function from $V_{n}^{(p)}$ to $V_{m}^{(p)}$ satisfying $G_{\alpha_{i}}=(F_{\alpha_{i}})^{*}, 1\leq i \leq m$. For any $c \in {V_{m}^{(p)}}^{*}$, denote $c$ by $c=a_{i_{1}}\alpha_{i_{1}}+\dots+a_{i_{t}}\alpha_{i_{t}}$, where $a_{i_{j}} \in \mathbb{F}_{p}^{*}, 1\leq j \leq t$. Then by Eq. (20) we have
\begin{equation*}
\begin{split}
G_{c}(x)&=\langle c, G(x)\rangle_{m}\\
&=a_{i_{1}}\langle \alpha_{i_{1}}, G(x)\rangle_{m}+\dots+a_{i_{t}}\langle \alpha_{i_{t}}, G(x)\rangle_{m}\\
&=a_{i_{1}}G_{\alpha_{i_{1}}}(x)+\dots+a_{i_{t}}G_{\alpha_{i_{t}}}(x)\\
&=a_{i_{1}}(F_{\alpha_{i_{1}}})^{*}(x)+\dots+a_{i_{t}}(F_{\alpha_{i_{t}}})^{*}(x)\\
&=(F_{a_{i_{1}}\alpha_{i_{1}}})^{*}(x)+\dots+(F_{a_{i_{t}}\alpha_{i_{t}}})^{*}(x)\\
&=(F_{a_{i_{1}}\alpha_{i_{1}}+\dots+a_{i_{t}}\alpha_{i_{t}}})^{*}(x)\\
&=(F_{c})^{*}(x).
\end{split}
\end{equation*}
For any $c \in {V_{m}^{(p)}}^{*}$, since $G_{c}=(F_{c})^{*}$ is bent, we have that $G$ is vectorial bent. Therefore, $F$ is vectorial dual-bent with $(F_{c})^{*}=(F^{*})_{c}, c \in {V_{m}^{(p)}}^{*}$, where $F^{*}=G$.
\end{proof}

Below we give an example to illustrate Theorem 4.

\begin{example}\label{Example 2}
Let $F: \mathbb{F}_{3^6} \times \mathbb{F}_{3^6}\rightarrow \mathbb{F}_{3^2}$ be defined by
\begin{equation*}
F(x_{1}, x_{2})=Tr_{2}^{6}(x_{1}x_{2}^{717}).
\end{equation*}
Then by Proposition 3 of \emph{\cite{WFW2023Be}}, $F$ is a vectorial dual-bent function with Condition A and the corresponding $\varepsilon=1$. For any $c \in \mathbb{F}_{3^2}^{*}, z=(z_{1}, z_{2}) \in \mathbb{F}_{3^6} \times \mathbb{F}_{3^6}$, define
\begin{equation*}
H_{c}^{(z)}=\left[\zeta_{3}^{Tr_{1}^{6}(c(x_{1}-y_{1})(x_{2}-y_{2})^{717})-Tr_{1}^{6}(z_{1}(x_{1}-y_{1})+z_{2}(x_{2}-y_{2}))}\right]_{(x_{1}, x_{2}), (y_{1}, y_{2}) \in \mathbb{F}_{3^6} \times \mathbb{F}_{3^6}}.
\end{equation*}
Denote $H_{c}^{(0)}$ by $H_{c}$. Then by Theorem 4, $H_{c}^{(z)}, c \in \mathbb{F}_{3^2}^{*}, z=(z_{1}, z_{2}) \in \mathbb{F}_{3^6} \times \mathbb{F}_{3^6}$ are all generalized Hadamard matrices for which
\begin{equation*}
729^{-1}\sum_{i \in \mathbb{F}_{3^6} \times \mathbb{F}_{3^6}}(H_{c}^{(z)})_{i, (0, 0)} \in \{1, \zeta_{3}, \zeta_{3}^{2}\} \ \text{ for all } \ c \in \mathbb{F}_{3^2}^{*}, z=(z_{1}, z_{2}) \in \mathbb{F}_{3^6} \times \mathbb{F}_{3^6},
\end{equation*}
and $H_{c}\overline{H_{d}}^{\top}=729H_{c-d}$ for any $c\neq d \in \mathbb{F}_{3^2}^{*}$.
\end{example}

\section{A characterization of vectorial dual-bent functions with Condition A in terms of bent partitions when $p=2$} \label{sec: 6}

When $p$ is an odd prime, a characterization of vectorial dual-bent functions with Condition A in terms of bent partitions has been given in \cite{WFW2023Be}, see Lemma 2. In this section, when $p=2$, we give a characterization of vectorial dual-bent functions with Condition A in terms of bent partitions. First, we give a lemma, which characterizes bent partitions of $V_{n}^{(2)}$ of depth $2^m$ in terms of vectorial bent functions.

\begin{lemma}\label{Lemma 8}
Let $\Gamma=\{A_{i}, i \in V_{m}^{(2)}\}$ be a partition of $V_{n}^{(2)}$, where $n\geq 4$ is even, $2\leq m\leq \frac{n}{2}$. Define $F: V_{n}^{(2)} \rightarrow V_{m}^{(2)}$ as
\begin{equation*}
F(x)=\sum_{i \in V_{m}^{(2)}}\delta_{A_{i}}(x)i.
\end{equation*}
The following two statements are equivalent.

\emph{(1)} $\Gamma$ is a bent partition.

\emph{(2)} $F$ is a vectorial bent function for which there exists a function $G: V_{n}^{(2)}\rightarrow V_{m}^{(2)}$ and a set $S\subseteq V_{n}^{(2)}$ such that
\begin{equation*}
(F_{c})^{*}(x)=G_{c}(x)+\delta_{S}(x), c \in {V_{m}^{(2)}}^{*}, x \in V_{n}^{(2)}.
\end{equation*}
\end{lemma}

\begin{proof}
(1) $\Rightarrow$ (2): By the result in \cite{Dillon1974El}, for any Boolean bent function $f: V_{n}^{(2)}\rightarrow \mathbb{F}_{2}$ and $u \in V_{n}^{(2)}, j \in \mathbb{F}_{2}$, we have
\begin{equation}\label{21}
\chi_{u}(D_{f, j})=\left\{
\begin{split}
2^{n-1}\delta_{0}(u)+2^{\frac{n}{2}-1}, & \ \text{ if } f^{*}(u)=j, \\
2^{n-1}\delta_{0}(u)-2^{\frac{n}{2}-1}, & \ \text{ if } f^{*}(u)=j+1.
\end{split}\right.
\end{equation}
By Eq. (21) and the definition of bent partitions, we have that for any fixed $u \in V_{n}^{(2)}$,
\begin{equation}\label{22}
\chi_{u}(D_{F, i})=\chi_{u}(A_{i})=\left\{
\begin{split}
2^{n-m}\delta_{0}(u)-2^{\frac{n}{2}-m}, & \ \text{ if } i \in V_{m}^{(2)}\backslash \{G(u)\}, \\
2^{n-m}\delta_{0}(u)-2^{\frac{n}{2}-m}+2^{\frac{n}{2}}, & \ \text{ if } i=G(u),
\end{split}\right.
\end{equation}
or
\begin{equation}\label{23}
\chi_{u}(D_{F, i})=\chi_{u}(A_{i})=\left\{
\begin{split}
2^{n-m}\delta_{0}(u)+2^{\frac{n}{2}-m}, & \ \text{ if } i \in V_{m}^{(2)}\backslash \{G(u)\}, \\
2^{n-m}\delta_{0}(u)+2^{\frac{n}{2}-m}-2^{\frac{n}{2}}, & \ \text{ if } i=G(u),
\end{split}\right.
\end{equation}
where $G$ is some function from $V_{n}^{(2)}$ to $V_{m}^{(2)}$. Let
\begin{equation*}
S=\{u \in V_{n}^{(2)}: \chi_{u}(D_{F, i}), i \in V_{m}^{(2)} \text{ satisfy Eq. (23)}\}.
\end{equation*}
Then for any $c \in {V_{m}^{(2)}}^{*}$ and $u \in V_{n}^{(2)}$ we obtain
\begin{equation*}
\begin{split}
W_{F_{c}}(u)&=\sum_{x \in V_{n}^{(2)}}(-1)^{\langle c, F(x)\rangle_{m}+\langle u, x\rangle_{n}}\\
&=\sum_{i \in V_{m}^{(2)}}\sum_{x \in V_{n}^{(2)}: F(x)=i}(-1)^{\langle c, F(x)\rangle_{m}+\langle u, x\rangle_{n}}\\
&=\sum_{i \in V_{m}^{(2)}}(-1)^{\langle c, i\rangle_{m}}\sum_{x \in V_{n}^{(2)}: F(x)=i}(-1)^{\langle u, x\rangle_{n}}\\
&=\sum_{i \in V_{m}^{(2)}}(-1)^{\langle c, i\rangle_{m}}\chi_{u}(D_{F, i})\\
&=\left\{
\begin{split}
\sum_{i \in V_{m}^{(2)}}(-1)^{\langle c, i\rangle_{m}}(2^{n-m}\delta_{0}(u)+2^{\frac{n}{2}-m}-2^{\frac{n}{2}}\delta_{G(u)}(i)), & \ \text{ if } u \in S,\\
\sum_{i \in V_{m}^{(2)}}(-1)^{\langle c, i\rangle_{m}}(2^{n-m}\delta_{0}(u)-2^{\frac{n}{2}-m}+2^{\frac{n}{2}}\delta_{G(u)}(i)), & \ \text { if } u \notin S,
\end{split}\right.\\
&=\left\{
\begin{split}
(2^{n-m}\delta_{0}(u)+2^{\frac{n}{2}-m})\sum_{i \in V_{m}^{(2)}}(-1)^{\langle c, i\rangle_{m}}-2^{\frac{n}{2}}(-1)^{\langle c, G(u)\rangle_{m}}, & \ \text{ if } u \in S,\\
(2^{n-m}\delta_{0}(u)-2^{\frac{n}{2}-m}) \sum_{i \in V_{m}^{(2)}}(-1)^{\langle c, i\rangle_{m}}+2^{\frac{n}{2}}(-1)^{\langle c, G(u)\rangle_{m}}, & \ \text { if } u \notin S,
\end{split}\right.\\
&=\left\{
\begin{split}
2^{\frac{n}{2}}(-1)^{1+\langle c, G(u)\rangle_{m}}, & \ \text{ if } u \in S,\\
2^{\frac{n}{2}}(-1)^{\langle c, G(u)\rangle_{m}}, & \ \text { if } u \notin S,
\end{split}\right.
\end{split}
\end{equation*}
which implies that $F$ is a vectorial bent function with $(F_{c})^{*}(x)=G_{c}(x)+\delta_{S}(x), c \in {V_{m}^{(2)}}^{*}, x \in V_{n}^{(2)}$.

(2) $\Rightarrow$ (1): With the same argument as in the proof of Proposition 3 of \cite{WF2023Ne}, for any $u \in V_{n}^{(2)}, i \in V_{m}^{(2)}$ we have
\begin{equation}\label{24}
     \chi_{u}(D_{F, i})=2^{n-m}\delta_{0}(u)+2^{-m}\sum_{c \in {V_{m}^{(2)}}^{*}}W_{F_{c}}(u)(-1)^{\langle c, i\rangle_{m}}.
\end{equation}
Since $F$ is a vectorial bent function with $(F_{c})^{*}(x)=G_{c}(x)+\delta_{S}(x), c \in {V_{m}^{(2)}}^{*}$, by Eq. (24) we have
\begin{equation}\label{25}
\begin{split}
\chi_{u}(A_{i})
&=\chi_{u}(D_{F, i})\\
&=2^{n-m}\delta_{0}(u)+2^{\frac{n}{2}-m}\sum_{c \in {V_{m}^{(2)}}^{*}}(-1)^{(F_{c})^{*}(u)+\langle c, i\rangle_{m}}\\
&=2^{n-m}\delta_{0}(u)+2^{\frac{n}{2}-m}\sum_{c \in {V_{m}^{(2)}}^{*}}(-1)^{G_{c}(u)+\delta_{S}(u)+\langle c, i\rangle_{m}} \\
&=2^{n-m}\delta_{0}(u)+(-1)^{\delta_{S}(u)}2^{\frac{n}{2}-m}\sum_{c \in {V_{m}^{(2)}}^{*}}(-1)^{\langle c,G(u)+i\rangle_{m}} \\ &=2^{n-m}\delta_{0}(u)+(-1)^{\delta_{S}(u)}2^{\frac{n}{2}-m}(2^{m}\delta_{G(u)}(i)-1).\\
\end{split}
\end{equation}
For any union $D$ of $2^{m-1}$ sets of $\{A_{i}, i \in V_{m}^{(2)}\}$, we have
\begin{equation}\label{26}
 \chi_{u}(D)=\left\{\begin{split}
                      2^{n-1}\delta_{0}(u)+(-1)^{\delta_{S}(u)}2^{\frac{n}{2}-1},  & \ \text{ if } A_{G(u)}\subseteq D,\\
                      2^{n-1}\delta_{0}(u)-(-1)^{\delta_{S}(u)}2^{\frac{n}{2}-1},  & \ \text{ if } A_{G(u)}\nsubseteq D.
                    \end{split}
 \right.
\end{equation}
Let $f: V_{n}^{(2)}\rightarrow \mathbb{F}_{2}$ be a function for which for each $j \in \mathbb{F}_{2}$, there are exactly $2^{m-1}$ sets $A_{i}$ in $\Gamma$ in its preimage set. By Eq. (26), for any $u \in V_{n}^{(2)}$ we have
\begin{equation*}
  \chi_{u}(D_{f, j})=\left\{\begin{split}
                               2^{n-1}\delta_{0}(u)+(-1)^{\delta_{S}(u)}2^{\frac{n}{2}-1}, & \ \text{ if } j=g(u), \\
                               2^{n-1}\delta_{0}(u)-(-1)^{\delta_{S}(u)}2^{\frac{n}{2}-1},  & \ \text{ if } j=g(u)+1,
                            \end{split}\right.
\end{equation*}
where $g(u)=f(A_{G(u)})$. Then we obtain
\begin{equation*}
\begin{split}
W_{f}(u)&=\sum_{x \in V_{n}^{(2)}}(-1)^{f(x)+\langle u, x\rangle_{n}}\\
&=\sum_{j \in \mathbb{F}_{2}}\sum_{x \in V_{n}^{(2)}: f(x)=j}(-1)^{f(x)+\langle u, x\rangle_{n}}\\
&=\sum_{j \in \mathbb{F}_{2}}(-1)^{j}\sum_{x \in V_{n}^{(2)}: f(x)=j}(-1)^{\langle u, x\rangle_{n}}\\
&=\sum_{j \in \mathbb{F}_{2}}(-1)^{j}\chi_{u}(D_{f, j})\\
\end{split}
\end{equation*}
\begin{equation*}
\begin{split}
&=(2^{n-1}\delta_{0}(u)+(-1)^{\delta_{S}(u)}2^{\frac{n}{2}-1})(-1)^{g(u)}+(2^{n-1}\delta_{0}(u)-(-1)^{\delta_{S}(u)}2^{\frac{n}{2}-1})(-1)^{g(u)+1}\\
&=2^{\frac{n}{2}}(-1)^{g(u)+\delta_{S}(u)},
\end{split}
\end{equation*}
which implies that $f$ is a Boolean bent function, and thus $\Gamma$ is a bent partition.
\end{proof}

The following theorem gives a characterization of vectorial dual-bent functions $F: V_{n}^{(2)}\rightarrow V_{m}^{(2)}$ with Condition A in terms of bent partitions.

\begin{theorem}\label{Theorem 5}
Let $F: V_{n}^{(2)}\rightarrow V_{m}^{(2)}$, where $n\geq 4$ is even and $2\leq m \leq \frac{n}{2}$. The following two statements are equivalent.

\emph{(1)} $F$ is a vectorial dual-bent function with $(F_{c})^{*}=(F^{*})_{c}, c \in {V_{m}^{(2)}}^{*}$.

\emph{(2)} $\Gamma=\{D_{F, i}, i \in V_{m}^{(2)}\}$ is a bent partition of $V_{n}^{(2)}$ with $\chi_{u}(D_{F, i}) \in \{-2^{\frac{n}{2}-m}, -2^{\frac{n}{2}-m}+2^{\frac{n}{2}}\}, u \in {V_{n}^{(2)}}^{*}, i \in V_{m}^{(2)}$.
\end{theorem}

\begin{proof}
By the proof of Lemma 8, $F$ is a vectorial dual-bent function with $(F_{c})^{*}=(F^{*})_{c}, c \in {V_{m}^{(2)}}^{*}$ if and only if $\Gamma=\{D_{F, i}, i \in V_{m}^{(2)}\}$ is a bent partition with $\chi_{u}(D_{F, i}) \in \{2^{n-m}\delta_{0}(u)-2^{\frac{n}{2}-m}, 2^{n-m}\delta_{0}(u)-2^{\frac{n}{2}-m}+2^{\frac{n}{2}}\}, u \in V_{n}^{(2)}, i \in V_{m}^{(2)}$. In the following, we only need to show that when $\Gamma=\{D_{F, i}, i \in V_{m}^{(2)}\}$ is a bent partition and $\chi_{u}(D_{F, i}) \in \{-2^{\frac{n}{2}-m}, -2^{\frac{n}{2}-m}+2^{\frac{n}{2}}\}, u \in {V_{n}^{(2)}}^{*}, i \in V_{m}^{(2)}$, then $\chi_{0}(D_{F, i})=|D_{F, i}| \in \{2^{n-m}-2^{\frac{n}{2}-m}, 2^{n-m}-2^{\frac{n}{2}-m}+2^{\frac{n}{2}}\}, i \in V_{m}^{(2)}$.

For any $i \in V_{m}^{(2)}$, let $b_{i}=|\{u \in {V_{n}^{(2)}}^{*}: \chi_{u}(D_{F, i})=-2^{\frac{n}{2}-m}+2^{\frac{n}{2}}\}|$. Assume that there is $i$ such that $|D_{F, i}| \notin \{2^{n-m}-2^{\frac{n}{2}-m}, 2^{n-m}-2^{\frac{n}{2}-m}+2^{\frac{n}{2}}\}$. Then by Lemma 1, there exists $i_{0}$ such that $|D_{F, i_{0}}|=2^{n-m}+2^{\frac{n}{2}-m}-2^{\frac{n}{2}}, |D_{F, i}|=2^{n-m}+2^{\frac{n}{2}-m}, i \neq i_{0}$. Let $j \in V_{m}^{(2)}$ with $j\neq F(0), i_{0}$. Then $0 \notin D_{F, j}$ and $|D_{F, j}|=2^{n-m}+2^{\frac{n}{2}-m}$. Since
\begin{equation*}
\sum_{u \in V_{n}^{(2)}}\chi_{u}(D_{F, j})=\sum_{u \in V_{n}^{(2)}}\sum_{x \in D_{F, j}}(-1)^{\langle u, x\rangle_{n}}=\sum_{x \in D_{F, j}}\sum_{u \in V_{n}^{(2)}}(-1)^{\langle u, x\rangle_{n}}=2^{n}\delta_{D_{F, j}}(0)=0,
\end{equation*}
and
\begin{equation*}
\begin{split}
\sum_{u \in V_{n}^{(2)}}\chi_{u}(D_{F, j})&=|D_{F, j}|+(-2^{\frac{n}{2}-m}+2^{\frac{n}{2}})b_{j}-2^{\frac{n}{2}-m}(2^n-1-b_{j})\\
&=2^{\frac{n}{2}-m}(2^{\frac{n}{2}}-2^{n}+2+2^{m}b_{j}),
\end{split}
\end{equation*}
we have
\begin{equation}\label{27}
2^{n}=2^{\frac{n}{2}}+2^{m}b_{j}+2.
\end{equation}
Note that $b_{j}\neq 0$ by $n\geq 4$. Since $m\leq \frac{n}{2}$, we have $2^{m} \mid 2^{n}, 2^{m} \mid 2^{\frac{n}{2}}, 2^{m} \mid 2^{m}b_{j}$. Thus by Eq. (27), $2^m \mid 2$, which contradicts $m\geq 2$. Therefore, $|D_{F, i}| \in \{2^{n-m}-2^{\frac{n}{2}-m}, 2^{n-m}-2^{\frac{n}{2}-m}+2^{\frac{n}{2}}\}, i \in V_{m}^{(2)}$.
\end{proof}

Below we give an example to illustrate Theorem 5.
\begin{example}\label{Example 3}
Let $F: \mathbb{F}_{2^6} \times \mathbb{F}_{2^6} \times \mathbb{F}_{2^4} \times \mathbb{F}_{2^4}\rightarrow \mathbb{F}_{2^2}$ be defined by
\begin{equation*}
F(x_{1}, x_{2}, x_{3}, x_{4})=(Tr_{2}^{4}(x_{3}x_{4}^{14}))^{3}Tr_{2}^{6}(x_{1}^{52}x_{2}-x_{1}x_{2}^{58})+Tr_{2}^{6}(x_{1}x_{2}^{58})+Tr_{2}^{4}(\alpha x_{3}x_{4}^{14}),
\end{equation*}
where $\alpha$ is a primitive element of $\mathbb{F}_{2^4}$. By Theorem 5 of \emph{\cite{WFW2023Be}}, $F$ is a vectorial dual-bent function with Condition A. By Theorem 5, $\{D_{F, i}, i \in \mathbb{F}_{2^2}\}$ is a bent partition with $\chi_{u}(D_{F, i})\in \{-256, 768\}, u \in (\mathbb{F}_{2^6} \times \mathbb{F}_{2^6} \times \mathbb{F}_{2^4} \times \mathbb{F}_{2^4})^{*}, i \in \mathbb{F}_{2^2}$.
\end{example}

\section{New characterizations of certain bent partitions} \label{sec: 7}

In this section, we give some new characterizations of bent partitions with Condition $\mathcal{C}$ when $p$ is an odd prime, and bent partitions with condition given in Theorem 5 when $p=2$, respectively.

\begin{theorem}\label{Theorem 6}
Let $p$ be an odd prime. Let $\Gamma=\{A_{i}, i \in V_{m}^{(p)} \}$ be a partition of $V_{n}^{(p)}$, where $n\geq 4$ is even and $2\leq m \leq \frac{n}{2}$. Denote $0 \in A_{i_{0}}$ and $I=\{\sum_{i \in V_{m}^{(p)}}\delta_{A_{i}}(x)i: x \in {V_{n}^{(p)}}^{*}\}$. The following statements are pairwise equivalent.

\emph{(1)} $\Gamma$ is a bent partition with Condition $\mathcal{C}$.

\emph{(2)} For any $i \in V_{m}^{(p)}$, $A_{i}^{*}$ is a regular $(p^{n}, s_{i}(p^{\frac{n}{2}}-\varepsilon), \varepsilon p^{\frac{n}{2}}+s_{i}^{2}-3\varepsilon s_{i}, s_{i}^{2}-\varepsilon s_{i})$ partial difference set, where $s_{i}=p^{\frac{n}{2}-m}+\varepsilon \delta_{i_{0}}(i)$, $\varepsilon \in \{\pm 1\}$ is a constant with $\varepsilon=1$ if $p\neq 3$.

\emph{(3)} Let
\begin{equation*}
\begin{split}
&R_{id}=\{(x, x): x \in V_{n}^{(p)}\},\\
&R_{i}=\{(x, y): x, y \in V_{n}^{(p)}, x-y \in A_{i}^{*}\}, i \in I.
\end{split}
\end{equation*}
Then $\{R_{id}, R_{i}, i \in I\}$ is an $|I|$-class amorphic association scheme for which $|I|\geq 3$ and for any $i \in I$, the intersection number $p_{i, i}^{id}=p^{n-m}-\varepsilon p^{\frac{n}{2}-m}+\delta_{i_{0}}(i)(\varepsilon p^{\frac{n}{2}}-1)$, where $\varepsilon \in \{\pm 1\}$ is a constant with $\varepsilon=1$ if $p\neq 3$.

\emph{(4)} Let
\begin{equation*}
C_{\widetilde{A_{i}^{*}}}=\{c_{\alpha}=(\langle \alpha, x\rangle_{n})_{x \in \widetilde{A_{i}^{*}}}: \alpha \in V_{n}^{(p)}\}, i \in I,
\end{equation*}
where $\widetilde{A_{i}^{*}}$ is a subset of $A_{i}^{*}$ for which any two elements in $\widetilde{A_{i}^{*}}$ are linearly independent and for any $x \in A_{i}^{*}$, there exist $a \in \mathbb{F}_{p}^{*}, x' \in \widetilde{A_{i}^{*}}$ such that $x=ax'$. Then for any $i \in I$, $C_{\widetilde{A_{i}^{*}}}$ is a two-weight $[\frac{p^{n-m}-\varepsilon p^{\frac{n}{2}-m}+\delta_{i_{0}}(i)(\varepsilon p^{\frac{n}{2}}-1)}{p-1}, n]$ projective linear code and the two nonzero weights are
\begin{equation*}
\begin{split}
& w_{1}=p^{n-m-1}+\frac{1-\varepsilon +2\varepsilon \delta_{i_{0}}(i)}{2}p^{\frac{n}{2}-1},\\
& w_{2}=p^{n-m-1}+\frac{-1-\varepsilon +2\varepsilon \delta_{i_{0}}(i)}{2}p^{\frac{n}{2}-1},
\end{split}
\end{equation*}
where $\varepsilon \in \{\pm 1\}$ is a constant with $\varepsilon=1$ if $p\neq 3$.

\emph{(5)} Let
\begin{equation*}
H_{c}^{(z)}=\left[\zeta_{p}^{\langle c, \sum_{i \in V_{m}^{(p)}}\delta_{A_{i}}(x-y)i\rangle_{m}- \langle z, x-y\rangle_{n}}\right]_{x, y \in V_{n}^{(p)}}, c \in {V_{m}^{(p)}}^{*}, z \in V_{n}^{(p)},
\end{equation*}
and $H_{c}=H_{c}^{(0)}$.
Then $H_{c}^{(z)}, c \in {V_{m}^{(p)}}^{*}, z \in V_{n}^{(p)}$ are all generalized Hadamard matrices for which there exists a constant $\varepsilon \in \{\pm 1\}$ with $\varepsilon=1$ if $p\neq 3$ such that \emph{Eq. (15)} holds and $H_{c}\overline{H_{d}}^{\top}=\varepsilon p^{\frac{n}{2}}H_{c-d}$ for any $c\neq d \in {V_{m}^{(p)}}^{*}$.

Furthermore, if any one of the above statements holds, then $I=V_{m}^{(p)}$ and $|I|=p^m$ except one case that $p=3$, $n=2m$ and $\varepsilon=-1$ (in such a case, $I=V_{m}^{(3)} \backslash \{i_{0}\}$ and $|I|=3^m-1$).
\end{theorem}

\begin{proof}
By Lemma 2, statement (1) holds if and only if $F: V_{n}^{(p)}\rightarrow V_{m}^{(p)}$ is a vectorial dual-bent function with Condition A, where
\begin{equation*}
F(x)=\sum_{i \in V_{m}^{(p)}}\delta_{A_{i}}(x)i.
\end{equation*}
Then the result follows from Proposition 3 and Theorems 1, 2, 4.
\end{proof}

\begin{theorem}\label{Theorem 7}
Let $\Gamma=\{A_{i}, i \in V_{m}^{(2)} \}$ be a partition of $V_{n}^{(2)}$, where $n\geq 4$ is even and $2\leq m \leq \frac{n}{2}$. Denote $0 \in A_{i_{0}}$. The following statements are pairwise equivalent.

\emph{(1)} $\Gamma$ is a bent partition with $\chi_{u}(A_{i}) \in \{-2^{\frac{n}{2}-m}, -2^{\frac{n}{2}-m}+2^{\frac{n}{2}}\}, u \in {V_{n}^{(2)}}^{*}, i \in V_{m}^{(2)}$.

\emph{(2)} For any $i \in V_{m}^{(2)}$, $A_{i}^{*}$ is a regular $(2^{n}, s_{i}(2^{\frac{n}{2}}-1), 2^{\frac{n}{2}}+s_{i}^{2}-3 s_{i}, s_{i}^{2}-s_{i})$ partial difference set, where $s_{i}=2^{\frac{n}{2}-m}+\delta_{i_{0}}(i)$.

\emph{(3)} Let
\begin{equation*}
\begin{split}
& R_{id}=\{(x, x): x \in V_{n}^{(2)}\}, \\
& R_{i}=\{(x, y): x, y \in V_{n}^{(2)}, x+y \in A_{i}^{*}\}, i \in V_{m}^{(2)}.
\end{split}
\end{equation*}
Then $\{R_{id}, R_{i}, i \in V_{m}^{(2)}\}$ is a $2^m$-class amorphic association scheme for which for any $i \in V_{m}^{(2)}$, the intersection number $p_{i, i}^{id}=2^{n-m}-2^{\frac{n}{2}-m}+\delta_{i_{0}}(i)(2^{\frac{n}{2}}-1)$.

\emph{(4)} Let
\begin{equation*}
C_{A_{i}^{*}}=\{c_{\alpha}=(\langle \alpha, x\rangle_{n})_{x \in A_{i}^{*}}: \alpha \in V_{n}^{(2)}\}, i \in V_{m}^{(2)}.
\end{equation*}
Then for any $i \in V_{m}^{(2)}$, $C_{A_{i}^{*}}$ is a two-weight $[2^{n-m}- 2^{\frac{n}{2}-m}+\delta_{i_{0}}(i)(2^{\frac{n}{2}}-1), n]$ projective linear code and the two nonzero weights are
\begin{equation*}
\begin{split}
& w_{1}=2^{n-m-1},\\
& w_{2}=2^{n-m-1}-2^{\frac{n}{2}-1}+\delta_{i_{0}}(i)2^{\frac{n}{2}}.
\end{split}
\end{equation*}

\emph{(5)} Let
\begin{equation*}
H_{c}=\left[(-1)^{\langle c, \sum_{i \in V_{m}^{(2)}}\delta_{A_{i}}(x+y)i\rangle_{m}}\right]_{x, y \in V_{n}^{(2)}}, c \in {V_{m}^{(2)}}^{*}.
\end{equation*}
Then $H_{c}, c \in {V_{m}^{(2)}}^{*}$ are all Hadamard matrices and $H_{c}H_{d}=2^{\frac{n}{2}}H_{c+d}$ for any $c\neq d \in {V_{m}^{(2)}}^{*}$.
\end{theorem}

\begin{proof}
By Theorem 5, statement (1) holds if and only if $F: V_{n}^{(2)}\rightarrow V_{m}^{(2)}$ is a vectorial dual-bent function with $(F_{c})^{*}=(F^{*})_{c}, c \in {V_{m}^{(2)}}^{*}$, where
\begin{equation*}
F(x)=\sum_{i \in V_{m}^{(2)}}\delta_{A_{i}}(x)i.
\end{equation*}
Then the result follows from Corollaries 1, 2, 3 and Theorem 3.
\end{proof}

\begin{remark}\label{Remark 5}
As far as we know, the known bent partitions $\Gamma$ of $V_{n}^{(2)}$ of depth $2^m$ with $m\geq 2$ given in \emph{\cite{AKM2023Ge,MP2021Be,WFW2023Be}} all satisfy the statement \emph{(1)} of Theorem 7.
\end{remark}

\section{Association schemes from general vectorial dual-bent functions with $F(0)=0, F(x)=F(-x)$ and $2\leq m \leq \frac{n}{2}$} \label{sec: 8}

In \cite{AKMO2023Ve}, Anbar \emph{et al.} showed that vectorial dual-bent functions $F: V_{n}^{(p)} \rightarrow V_{m}^{(p)}$ with $F(0)=0$ and all component functions $F_{c}, c \in {V_{m}^{(p)}}^{*}$ being regular or weakly regular but not regular (that is, $\varepsilon_{F_{c}}, c \in {V_{m}^{(p)}}^{*}$ are all the same) can induce association schemes. Note that for such vectorial dual-bent functions, $n$ must be even and $m \leq \frac{n}{2}$. In this section, we give a necessary and sufficient condition on constructing association schemes from general vectorial dual-bent functions $F: V_{n}^{(p)} \rightarrow V_{m}^{(p)}$ with $F(0)=0, F(x)=F(-x)$ and $2\leq m \leq \frac{n}{2}$. Note that the known vectorial dual-bent functions $F$ all satisfy $F(x)=F(-x)$. Based on our result, more association schemes can be yielded from some vectorial dual-bent functions $F: V_{n}^{(p)} \rightarrow V_{m}^{(p)}$ for which $n$ can be odd, or $n$ is even and  $\varepsilon_{F_{c}}, c \in {V_{m}^{(p)}}^{*}$ are not all the same. First, we need two lemmas.

\begin{lemma}\label{Lemma 9}
Let $F: V_{n}^{(p)} \rightarrow V_{m}^{(p)}$ be a vectorial dual-bent function with $F(0)=0, F(x)=F(-x)$, and $F^{*}$ be a vectorial dual of $F$. Then $F^{*}$ is also a vectorial dual-bent function with $F^{*}(0)=0, F^{*}(x)=F^{*}(-x)$.
\end{lemma}

\begin{proof}
Since $F^{*}$ is a vectorial dual of $F$, $(F_{c})^{*}=(F^{*})_{\sigma(c)}, c \in {V_{m}^{(p)}}^{*}$ for some permutation $\sigma$ over ${V_{m}^{(p)}}^{*}$. Since $F$ is vectorial dual-bent, we have that $(F_{c})^{*}, c \in {V_{m}^{(p)}}^{*}$ are all bent functions. By Theorem 3.1 of \cite{OP2020Du}, for any $p$-ary bent function $f$ whose dual $f^{*}$ is also bent, $(f^{*})^{*}(x)=f(-x)$ holds. Thus, for any $c \in {V_{m}^{(p)}}^{*}$,
\begin{equation*}
((F^{*})_{c})^{*}(x)=((F_{\sigma^{-1}(c)})^{*})^{*}(x)=F_{\sigma^{-1}(c)}(-x)=F_{\sigma^{-1}(c)}(x),
\end{equation*}
which implies that $F^{*}$ is a vectorial dual-bent function and a vectorial dual of $F^{*}$ is $F$. When $p=2$, obviously $F^{*}(x)=F^{*}(-x)$, and by the proof of Corollary 2 and Proposition 5 of \cite{CMP2021Ve}, we have $F^{*}(0)=0$. When $p$ is an odd prime, for any $p$-ary bent function $f$ with $f(x)=0, f(x)=f(-x)$, by Proposition II. 1 of \cite{OP2022Tw}, $f^{*}(0)=0, f^{*}(x)=f^{*}(-x)$. Thus for any $c \in {V_{m}^{(p)}}^{*}$, from $F_{c}(0)=0, F_{c}(x)=F_{c}(-x)$, we have $(F^{*})_{\sigma (c)}(0)=(F_{c})^{*}(0)=0, (F^{*})_{\sigma(c)}(-x)=(F_{c})^{*}(-x)=(F_{c})^{*}(x)=(F^{*})_{\sigma(c)}(x)$, which implies that $F^{*}(0)=0, F^{*}(x)=F^{*}(-x)$.
\end{proof}

\begin{lemma}\label{Lemma 10}
Let $F: V_{n}^{(p)} \rightarrow V_{m}^{(p)}$ be a vectorial dual-bent function with $F(0)=0, F(x)=F(-x)$ and $2\leq m \leq \frac{n}{2}$, and $F^{*}$ be a vectorial dual of $F$. Denote $\varepsilon_{F_{c}}(0)=p^{-\frac{n}{2}}\zeta_{p}^{-(F_{c})^{*}(0)}W_{F_{c}}(0)$, $c \in {V_{m}^{(p)}}^{*}$. Then
\begin{itemize}
  \item When $m<\frac{n}{2}$, then $|F({V_{n}^{(p)}}^{*})|=|F^{*}({V_{n}^{(p)}}^{*})|=p^m$;
  \item When $n$ is even and $m=\frac{n}{2}$, if $\varepsilon_{F_{c}}(0)=-1$ for all $c \in {V_{m}^{(p)}}^{*}$, then $|F({V_{n}^{(p)}}^{*})|=|F^{*}({V_{n}^{(p)}}^{*})|=p^m-1$, and if $\varepsilon_{F_{c}}(0), c \in {V_{m}^{(p)}}^{*}$ are not all $-1$, then $|F({V_{n}^{(p)}}^{*})|=|F^{*}({V_{n}^{(p)}}^{*})|=p^m$.
\end{itemize}
\end{lemma}

\begin{proof}
By Proposition 3 of \cite{WF2023Ne} and Eq. (24), for any $u \in V_{n}^{(p)}, i \in V_{m}^{(p)}$ we have
\begin{equation}\label{28}
     |D_{F, i}^{*}|=p^{n-m}+p^{-m}\sum_{c \in {V_{m}^{(p)}}^{*}}W_{F_{c}}(0)\zeta_{p}^{-\langle c, i\rangle_{m}}-\delta_{0}(i).
\end{equation}

Since $F^{*}$ is a vectorial dual of $F$, $(F_{c})^{*}=(F^{*})_{\sigma(c)}, c \in {V_{m}^{(p)}}^{*}$ for some permutation $\sigma$ over ${V_{m}^{(p)}}^{*}$. For any $c \in {V_{m}^{(p)}}^{*}$, $W_{F_{c}}(0)=\varepsilon _{F_{c}}(0)p^{\frac{n}{2}}\zeta_{p}^{(F_{c})^{*}(0)}=\varepsilon _{F_{c}}(0)p^{\frac{n}{2}}\zeta_{p}^{(F^{*})_{\sigma(c)}(0)}$. Since $F$ is vectorial dual-bent with $F(0)=0, F(x)=F(-x)$, by Lemma 9 we have $F^{*}(x)=0$ and $W_{F_{c}}(0)=\varepsilon _{F_{c}}(0)p^{\frac{n}{2}}, c \in {V_{m}^{(p)}}^{*}$. By Eq. (28), for any $i \in V_{m}^{(p)}$ we have
\begin{equation}\label{29}
 |D_{F, i}^{*}|=p^{n-m}+p^{\frac{n}{2}-m}\sum_{c \in {V_{m}^{(p)}}^{*}}\varepsilon_{F_{c}}(0)\zeta_{p}^{-\langle c, i\rangle_{m}}-\delta_{0}(i).
\end{equation}
By Eq. (29), for any $i \in {V_{m}^{(p)}}^{*}$, if $|D_{F, i}^{*}|=0$, then
\begin{equation*}
|\sum_{c \in {V_{m}^{(p)}}^{*}}\varepsilon_{F_{c}}(0)\zeta_{p}^{-\langle c, i\rangle_{m}}|=p^{\frac{n}{2}}.
\end{equation*}
Since $m\leq \frac{n}{2}$, we have $|\sum_{c \in {V_{m}^{(p)}}^{*}}\varepsilon_{F_{c}}(0)\zeta_{p}^{-\langle c, i\rangle_{m}}|\leq p^m-1< p^{\frac{n}{2}}$. Hence, for any $i \in {V_{m}^{(p)}}^{*}$, $|D_{F, i}^{*}|\neq 0$. When $i=0$, by Eq. (29) we have that $|D_{F, 0}^{*}|=0$ if and only if
\begin{equation*}
p^{\frac{n}{2}-m}\sum_{c \in {V_{m}^{(p)}}^{*}}\varepsilon_{F_{c}}(0)=1-p^{n-m}.
\end{equation*}
When $n$ is odd, by Theorem 1 of \cite{CM2013A},
\begin{equation}\label{30}
\varepsilon_{F_{a c}}(0)=\varepsilon_{F_{c}}(0)\eta_{1}(a) \ \text{ for any } a \in \mathbb{F}_{p}^{*}, c \in {V_{m}^{(p)}}^{*},
\end{equation}
where $\eta_{1}$ denotes the quadratic character of $\mathbb{F}_{p}$. From Eq. (30) and $\sum_{a \in \mathbb{F}_{p}^{*}}\eta_{1}(a)=0$, we can obtain $\sum_{c \in {V_{m}^{(p)}}^{*}}\varepsilon_{F_{c}}(0)=0$. Thus, when $n$ is odd, $|D_{F, 0}^{*}|\neq 0$ and $|F({V_{n}^{(p)}}^{*})|=p^m$. When $n$ is even and $m<\frac{n}{2}$, $p \mid p^{\frac{n}{2}-m}\sum_{c \in {V_{m}^{(p)}}^{*}}\varepsilon_{F_{c}}(0)$ (Note that $\varepsilon_{F_{c}}(0) \in \{\pm 1\}$ when $n$ is even) and $p \nmid (1-p^{n-m})$, thus $|D_{F, 0}^{*}|\neq 0$ and $|F({V_{n}^{(p)}}^{*})|=p^m$. When $n$ is even, $m=\frac{n}{2}$ and $\varepsilon_{F_{c}}(0)=-1, c \in {V_{m}^{(p)}}^{*}$, $p^{\frac{n}{2}-m}\sum_{c \in {V_{m}^{(p)}}^{*}}\varepsilon_{F_{c}}(0)=1-p^{n-m}$, thus $|D_{F, 0}^{*}|=0$ and $|F({V_{n}^{(p)}}^{*})|=p^m-1$. When $n$ is even, $m=\frac{n}{2}$ and $\varepsilon_{F_{c}}(0), c \in {V_{m}^{(p)}}^{*}$ are not all $-1$, $p^{\frac{n}{2}-m}\sum_{c \in {V_{m}^{(p)}}^{*}}\varepsilon_{F_{c}}(0)$ $ \neq 1-p^{n-m}$, thus $|D_{F, 0}^{*}|\neq 0$ and $|F({V_{n}^{(p)}}^{*})|=p^m$. From the above arguments, we have that the result of Lemma 10 holds for $F$. By Lemma 9, we have that $F^{*}$ is also vectorial dual-bent with $F^{*}(0)=0, F^{*}(-x)=F^{*}(x)$. By Proposition 2 of \cite{WF2022On}, for any $p$-ary bent function $f: V_{n}^{(p)} \rightarrow \mathbb{F}_{p}$ which satisfies that $n$ is even, $f(x)=f(-x)$ and the dual $f^{*}$ is also bent, $\varepsilon_{f^{*}}(0)=\varepsilon_{f}(0)$ holds. When $n$ is even, since $F_{c}, c \in {V_{m}^{(p)}}^{*}$ are all bent with $F_{c}(x)=F_{c}(-x)$ and the duals $(F_{c})^{*}, c \in {V_{m}^{(p)}}^{*}$ are also bent, we have $\varepsilon_{(F^{*})_{c}}(0)=\varepsilon_{(F_{\sigma^{-1}(c)})^{*}}(0)=\varepsilon_{F_{\sigma^{-1}(c)}}(0), c \in {V_{m}^{(p)}}^{*}$ and $\{\varepsilon_{(F^{*})_{c}}(0), c \in {V_{m}^{(p)}}^{*}\}=\{\varepsilon_{F_{c}}(0), c \in {V_{m}^{(p)}}^{*}\}$. Therefore, the result of Lemma 10 also holds for $F^{*}$.
\end{proof}

The following theorem gives a necessary and sufficient condition on constructing association schemes from general vectorial dual-bent functions $F: V_{n}^{(p)} \rightarrow V_{m}^{(p)}$ with $F(0)=0,F(x)=F(-x)$ and $2\leq m \leq \frac{n}{2}$.

\begin{theorem}\label{Theorem 8}
Let $F: V_{n}^{(p)} \rightarrow V_{m}^{(p)}$ be a vectorial dual-bent function with $F(0)=0, F(x)=F(-x)$ and $2\leq m \leq \frac{n}{2}$, and $F^{*}$ be a vectorial dual of $F$. Denote $I=F({V_{n}^{(p)}}^{*})$ and $\varepsilon_{F_{c}}(x)=p^{-\frac{n}{2}}\zeta_{p}^{-(F_{c})^{*}(x)}W_{F_{c}}(x), c \in {V_{m}^{(p)}}^{*}, x \in V_{n}^{(p)}$. Define
\begin{equation*}
\begin{split}
& R_{id}=\{(x, x): x \in V_{n}^{(p)}\},\\
& R_{i}=\{(x, y): x, y \in V_{n}^{(p)}, x-y \in D_{F, i}^{*}\}, i \in I.
\end{split}
\end{equation*}
Then
\begin{itemize}
  \item $I=V_{m}^{(p)}$ and $|I|=p^m$ except one case that $n$ is even, $m=\frac{n}{2}$ and $\varepsilon_{F_{c}}(0)=-1, c \in {V_{m}^{(p)}}^{*}$ (in such a case, $I={V_{m}^{(p)}}^{*}$ and $|I|=p^m-1$).
  \item The following two statements are equivalent.

  \emph{(1)} $\{R_{id}, R_{i}, i \in I\}$ is an $|I|$-class association scheme.

  \emph{(2)} For any $\beta, \beta' \in {V_{n}^{(p)}}^{*}$ with $F^{*}(\beta)=F^{*}(\beta')$, $\varepsilon_{F_{c}}(\beta)=\varepsilon_{F_{c}}(\beta'), c \in {V_{m}^{(p)}}^{*}$.
\end{itemize}
\end{theorem}

\begin{proof}
By Lemma 10 and its proof, we have that $I=V_{m}^{(p)}$ and $|I|=p^m$ except one case that $n$ is even, $m=\frac{n}{2}$ and $\varepsilon_{F_{c}}(0)=-1, c \in {V_{m}^{(p)}}^{*}$ (in such a case, $I={V_{m}^{(p)}}^{*}$ and $|I|=p^m-1$).

Since $F^{*}$ is a vectorial dual of $F$, $(F_{c})^{*}=(F^{*})_{\sigma(c)}, c \in {V_{m}^{(p)}}^{*}$ for some permutation $\sigma$ over ${V_{m}^{(p)}}^{*}$. By
\begin{equation*}
W_{F_{c}}(x)=\varepsilon_{F_{c}}(x)p^{\frac{n}{2}}\zeta_{p}^{(F_{c})^{*}(x)}=\varepsilon_{F_{c}}(x)p^{\frac{n}{2}}\zeta_{p}^{(F^{*})_{\sigma(c)}(x)}, c \in {V_{m}^{(p)}}^{*}, x \in V_{n}^{(p)},
\end{equation*}
where $\varepsilon_{F_{c}}(x) \in \{\pm 1, \pm \sqrt{-1}\}$ with $\varepsilon_{F_{c}}(x)=1$ if $p=2$, we have that for any $\beta, \beta' \in {V_{n}^{(p)}}^{*}$,
\begin{equation}\label{31}
\begin{split}
& W_{F_{c}}(\beta)=W_{F_{c}}(\beta'), c \in {V_{m}^{(p)}}^{*}\\
& \Leftrightarrow \varepsilon_{F_{c}}(\beta)=\varepsilon_{F_{c}}(\beta'), (F^{*})_{\sigma(c)}(\beta)=(F^{*})_{\sigma(c)}(\beta'), c \in {V_{m}^{(p)}}^{*}\\
& \Leftrightarrow \varepsilon_{F_{c}}(\beta)=\varepsilon_{F_{c}}(\beta'), c \in {V_{m}^{(p)}}^{*}, F^{*}(\beta)=F^{*}(\beta').
\end{split}
\end{equation}
By Lemma 10, $|F({V_{n}^{(p)}}^{*})|=|F^{*}({V_{n}^{(p)}}^{*})|$. Therefore, by relation (31),
\begin{equation*}
\begin{split}
& |\{(W_{F_{c}}(\beta))_{c \in {V_{m}^{(p)}}^{*}}: \beta \in {V_{n}^{(p)}}^{*}\}|=|I|\\
& \Leftrightarrow \text{ for any } \beta, \beta' \in {V_{n}^{(p)}}^{*} \text{ with } F^{*}(\beta)=F^{*}(\beta'), \varepsilon_{F_{c}}(\beta)=\varepsilon_{F_{c}}(\beta'), c \in {V_{m}^{(p)}}^{*}.
\end{split}
\end{equation*}
By Theorem 2 of \cite{AKMO2023Ve} (Note that $I=F(V_{n}^{(p)})$ in Theorem 2 of \cite{AKMO2023Ve} should be corrected as $I=F({V_{n}^{(p)} }^{*})$), $\{R_{id}, R_{i}, i \in I\}$ is an $|I|$-class association scheme if and only if
\begin{equation*}
|\{(W_{F_{c}}(\beta))_{c \in {V_{m}^{(p)}}^{*}}: \beta \in {V_{n}^{(p)}}^{*}\}|=|I|.
\end{equation*}
Hence, $\{R_{id}, R_{i}, i \in I\}$ is an $|I|$-class association scheme if and only if for any $\beta, \beta' \in {V_{n}^{(p)}}^{*}$ with $F^{*}(\beta)=F^{*}(\beta')$, we have $\varepsilon_{F_{c}}(\beta)=\varepsilon_{F_{c}}(\beta'), c \in {V_{m}^{(p)}}^{*}$.
\end{proof}

The following corollary is directly from Theorem 8, which states that for a vectorial dual-bent function $F: V_{n}^{(p)}\rightarrow V_{m}^{(p)}$ with $F(0)=0, F(x)=F(-x)$ and $2\leq m \leq \frac{n}{2}$, association schemes can be induced from $F$ if $F_{c}$ is weakly regular for any $c \in {V_{m}^{(p)}}^{*}$.

\begin{corollary} \label{Corollary 4}
Let $F: V_{n}^{(p)}\rightarrow V_{m}^{(p)}$ be a vectorial dual-bent function with $F(0)=0, F(x)=F(-x)$ and $2\leq m \leq \frac{n}{2}$. Denote $I=F({V_{n}^{(p)}}^{*})$. Define
\begin{equation*}
\begin{split}
& R_{id}=\{(x, x): x \in V_{n}^{(p)}\},\\
& R_{i}=\{(x, y): x, y \in V_{n}^{(p)}, x-y \in D_{F, i}^{*}\}, i \in I.
\end{split}
\end{equation*}
If $F_{c}$ is weakly regular for any $c \in {V_{m}^{(p)}}^{*}$, then $\{R_{id}, R_{i}, i \in I\}$ is an $|I|$-class association scheme, where $I=V_{m}^{(p)}$ and $|I|=p^m$ except one case that $n$ is even, $m=\frac{n}{2}$ and $\varepsilon_{F_{c}}=-1, c \in {V_{m}^{(p)}}^{*}$ (in such a case, $I={V_{m}^{(p)}}^{*}$ and $|I|=p^m-1$).
\end{corollary}

By using two classes of vectorial dual-bent functions $F: V_{n}^{(p)}\rightarrow V_{m}^{(p)}$ given in \cite{CMP2018Ve, WF2023Ne} for which $n$ can be odd, or $n$ is even and $\varepsilon_{F_{c}}, c \in {V_{m}^{(p)}}^{*}$ are not all the same, we can obtain more association schemes.

\begin{corollary}\label{5}
Let $p$ be an odd prime. Let $F: \mathbb{F}_{p^n} \rightarrow \mathbb{F}_{p^m}$ be defined as $F(x)=Tr_{m}^{n}(\alpha x^{2})$, where $m\geq 2, m \mid n, m\neq n$. Denote $I=F(\mathbb{F}_{p^n}^{*})$. Define
\begin{equation*}
\begin{split}
& R_{id}=\{(x, x): x \in \mathbb{F}_{p^n}\},\\
& R_{i}=\{(x, y): x, y \in \mathbb{F}_{p^n}, x-y \in D_{F, i}^{*}\}, i \in I.
\end{split}
\end{equation*}
Then $\{R_{id}, R_{i}, i \in I\}$ is an $|I|$-class association scheme, where $I=\mathbb{F}_{p^m}$ and $|I|=p^m$ except one case that $n$ is even, $m=\frac{n}{2}$ and $\eta_{n}(\alpha)=\xi^{n}$ (in such a case, $I=\mathbb{F}_{p^m}^{*}$ and $|I|=p^m-1$), where $\eta_{n}$ denotes the quadratic character of $\mathbb{F}_{p^n}$, $\xi=1$ if $p\equiv 1 \pmod 4$ and $\xi=\sqrt{-1}$ if $p \equiv 3 \pmod 4$.
\end{corollary}

\begin{proof}
Obviously, $F(0)=0, F(x)=F(-x)$. By Example 1 of \cite{CMP2018Ve}, $F$ is a vectorial dual-bent function for which for any $c \in \mathbb{F}_{p^m}^{*}$, the component function $F_{c}(x)=Tr_{1}^{n}(\alpha c x^{2})$ is weakly regular with $\varepsilon_{F_{c}}=(-1)^{n-1}\xi^{n}\eta_{n}(\alpha c)$. Then the result follows from Corollary 4.
\end{proof}

Below we give an example to illustrate Corollary 5.
\begin{example}\label{4}
Let $p=3, n=6, m=2$. Define $F(x)=Tr_{2}^{6}(x^{2}), x \in \mathbb{F}_{3^6}$. Then $F$ is a vectorial dual-bent function for which for any $c \in \mathbb{F}_{3^2}^{*}$, the component function $F_{c}(x)=Tr_{1}^{6}(c x^{2})$ is weakly regular with $\varepsilon_{F_{c}}=(-1)^{6-1}(\sqrt{-1})^{6}\eta_{6}(c)=\eta_{6}(c)$. Note that $\{\eta_{6}(c): c \in \mathbb{F}_{3^2}^{*}\}=\{\pm1\}$. Let
\begin{equation*}
\begin{split}
& R_{id}=\{(x, x): x \in \mathbb{F}_{3^6}\},\\
& R_{i}=\{(x, y): x, y \in \mathbb{F}_{3^6}, x-y \in D_{F, i}^{*}\}, i \in \mathbb{F}_{3^2}.
\end{split}
\end{equation*}
By Corollary 5, $\{R_{id}, R_{i}, i \in \mathbb{F}_{3^2}\}$ is a $9$-class association scheme.
\end{example}

\begin{corollary}\label{Corollary 6}
Let $p$ be an odd prime. Let $r_{1}, r_{2}, m$ be positive integers with $m\geq 2, m \mid r_{1}, m \mid r_{2}$. For $i \in \mathbb{F}_{p^m}$, define $H(i; x): \mathbb{F}_{p^{r_{1}}}\rightarrow \mathbb{F}_{p^{m}}$ as $H(0; x)=Tr_{m}^{r_{1}}(\alpha_{1} x^{2})$, $H(i; x)=Tr_{m}^{r_{1}}(\alpha_{2} x^{2})$ if $i$ is a square in $\mathbb{F}_{p^m}^{*}$, $H(i; x)=Tr_{m}^{r_{1}}(\alpha_{3} x^{2})$ if $i$ is a non-square in $\mathbb{F}_{p^m}^{*}$, where $\alpha_{j}, 1\leq j \leq 3$ are all squares or all non-squares in $\mathbb{F}_{p^{r_{1}}}^{*}$. Define $G: \mathbb{F}_{p^{r_{2}}} \times \mathbb{F}_{p^{r_{2}}} \rightarrow \mathbb{F}_{p^m}$ as $G(y_{1}, y_{2})=Tr_{m}^{r_{2}}(\beta y_{1} L(y_{2}))$, where $\beta \in \mathbb{F}_{p^{r_{2}}}^{*}$ and $L(x)=\sum a_{i} x^{q^{i}} (q=p^m)$ is a $q$-polynomial over $\mathbb{F}_{p^{r_{2}}}$ inducing a permutation of $\mathbb{F}_{p^{r_{2}}}$. Let $F: \mathbb{F}_{p^{r_{1}}} \times \mathbb{F}_{p^{r_{2}}} \times \mathbb{F}_{p^{r_{2}}}\rightarrow \mathbb{F}_{p^m}$ be defined as
\begin{equation*}
F(x, y_{1}, y_{2})=H(Tr_{m}^{r_{2}}(\gamma y_{2}^{2}); x)+G(y_{1}, y_{2}),
\end{equation*}
where $\gamma \in \mathbb{F}_{p^{r_{2}}}^{*}$. Define
\begin{equation*}
\begin{split}
& R_{id}=\{(x, x): x \in \mathbb{F}_{p^{r_{1}}} \times \mathbb{F}_{p^{r_{2}}} \times \mathbb{F}_{p^{r_{2}}}\},\\
& R_{i}=\{(x, y): x, y \in \mathbb{F}_{p^{r_{1}}} \times \mathbb{F}_{p^{r_{2}}} \times \mathbb{F}_{p^{r_{2}}}, x-y \in D_{F, i}^{*}\}, i \in \mathbb{F}_{p^m}.
\end{split}
\end{equation*}
Then $\{R_{id}, R_{i}, i \in \mathbb{F}_{p^m}\}$ is a $p^m$-class association scheme.
\end{corollary}

\begin{proof}
It is easy to see that $F(0, 0, 0)=0, F(x, y_{1}, y_{2})=F(-x, -y_{1}, -y_{2})$. By Theorem 1 of \cite{WF2023Ne} and its proof, $F$ is a vectorial dual-bent function for which for any $c \in \mathbb{F}_{p^m}^{*}$, the component function $F_{c}$ is weakly regular with $\varepsilon_{F_{c}}=(-1)^{r_{1}-1}\xi^{r_{1}}\eta_{r_{1}}(\alpha_{1} c)$, where $\eta_{r_{1}}$ denotes the quadratic character of $\mathbb{F}_{p^{r_{1}}}$, $\xi=1$ if $p\equiv 1 \pmod 4$ and $\xi=\sqrt{-1}$ if $p \equiv 3 \pmod 4$. Then the result follows from Corollary 4.
\end{proof}

Below we give an example to illustrate Corollary 6.
\begin{example}\label{5}
Let $p=5, r_{1}=r_{2}=9, m=3$ and $\alpha$ be a primitive element in $\mathbb{F}_{5^9}$. Then $n=r_{1}+2r_{2}=27$ is odd. For $i \in \mathbb{F}_{5^3}$, let $H(i; x)=Tr_{3}^{9}(x^{2}), x \in \mathbb{F}_{5^9}$ if $i=0$ and $H(i; x)=Tr_{3}^{9}(\alpha^{2} x^{2}), x \in \mathbb{F}_{5^9}$ if $i\in \mathbb{F}_{5^3}^{*}$. Define $F: V_{27}^{(5)}\rightarrow \mathbb{F}_{5^3}$ as $F(x, y_{1}, y_{2})=H(Tr_{3}^{9}(y_{2}^{2}); x)+Tr_{3}^{9}(y_{1}y_{2})=(Tr_{3}^{9}(y_{2}^{2}))^{124} Tr_{3}^{9}((\alpha^{2}-1)x^{2})+Tr_{3}^{9}(x^{2}+y_{1}y_{2})$, where $V_{27}^{(5)}=\mathbb{F}_{5^9} \times \mathbb{F}_{5^9} \times \mathbb{F}_{5^9}$. Let
\begin{equation*}
\begin{split}
& R_{id}=\{(x, x): x \in V_{27}^{(5)}\},\\
& R_{i}=\{(x, y): x, y \in V_{27}^{(5)}, x-y \in D_{F, i}^{*}\}, i \in \mathbb{F}_{5^3}.
\end{split}
\end{equation*}
By Corollary 6, $\{R_{id}, R_{i}, i \in \mathbb{F}_{5^3}\}$ is a $125$-class association scheme.
\end{example}

\section{Conclusion}
\label{sec: 6}
In this paper, we further studied vectorial dual-bent functions $F: V_{n}^{(p)} \rightarrow V_{m}^{(p)}$, where $2\leq m \leq \frac{n}{2}$. First, we gave new characterizations of vectorial dual-bent functions with Condition A in terms of amorphic association schemes (Theorem 1), linear codes (Theorem 2), generalized Hadamard matrices (Theorems 3 and 4), and bent partitions when $p=2$ (Theorem 5). Second, based on the relations between vectorial dual-bent functions with Condition A and bent partitions, new characterizations of certain bent partitions in terms of amorphic association schemes, linear codes and generalized Hadamard matrices were presented (Theorems 6 and 7). Finally, for general vectorial dual-bent functions $F: V_{n}^{(p)} \rightarrow V_{m}^{(p)}$ with $F(0)=0, F(x)=F(-x), 2\leq m \leq \frac{n}{2}$, we gave a necessary and sufficient condition on constructing association schemes (Theorem 8) and more association schemes were constructed (Corollaries 5 and 6).


\begin{thebibliography}{}
\bibitem{AKM2024Am}
N. Anbar, T.Kalayc{\i} and W. Meidl, Amorphic association schemes from bent partitions, \emph{Discret. Math.}, vol. 347, no. 1: 113658, 2024.
\bibitem{AKM2023Ge}
N. Anbar, T. Kalayc{\i} and W. Meidl, Generalized semifield spreads, \emph{Des. Codes Cryptogr.}, vol. 91, no. 2, pp. 545-562, 2023.
\bibitem{AKMO2023Ve}
N. Anbar, T.Kalayc{\i}, W. Meidl and F. \"{O}zbudak, Vectorial dual-bent functions and association schemes, Available: https://www.researchgate.net/publication/372046942
\bibitem{AM2022Be}
N. Anbar and W. Meidl, Bent partitions, \emph{Des. Codes Cryptogr.} vol. 90, no. 4, pp. 1081-1101, 2022.
\bibitem{CM2013A}
A. \c{C}e\c{s}melio\u{g}lu and W. Meidl, A construction of bent functions from plateaued functions, \emph{Des. Codes Cryptogr.} vol. 66, nos. 1-3, pp. 231-242, 2013.
\bibitem{CM2018Be}
A. \c{C}e\c{s}melio\u{g}lu and W. Meidl, Bent and vectorial bent functions, partial difference sets, and strongly regular graphs, \emph{Adv. Math. Commun.} vol. 12, pp. 691-705, 2018.
\bibitem{CMP2021Ve}
A. \c{C}e\c{s}melio\u{g}lu, W. Meidl and I. Pirsic, Vectorial bent functions and partial difference sets, \emph{Des. Codes Cryptogr.} vol. 89, no. 10, pp. 2313-2330, 2021.
\bibitem{CMP2018Ve}
A. \c{C}e\c{s}melio\u{g}lu, W. Meidl and A. Pott, Vectorial bent functions and their duals, \emph{Linear Algebra Appl.} vol. 548, pp. 305-320, 2018.
\bibitem{Dillon1974El}
J. F. Dillon, Elementary Hadamard difference sets, Ph.D. thesis, Department of Mathematics, University of Maryland, Colledge Park, MD, USA, 1974.
\bibitem{HLL2020Ra}
J. Y. Hyun, J. Lee and Y. Lee, Ramanujan graphs and expander families constructed from $p$-ary bent functions, \emph{Des. Codes Cryptogr.} vol. 88, no. 2,  pp. 453-470, 2020.
\bibitem{Ma1994A}
S. L. Ma, A survey of partial difference sets, \emph{Des. Codes Cryptogr.} vol. 4, no. 4, pp. 221-261, 1994.
\bibitem{MP2021Be}
W. Meidl and I. Pirsic, Bent and $\mathbb{Z}_{2^k}$-Bent functions from spread-like partitions, \emph{Des. Codes Cryptogr.}, vol. 89, no. 1, pp. 75-89, 2021.
\bibitem{Mesnager2016Be}
S. Mesnager, Bent Functions-Fundamentals and Results, Springer, Switzerland, 2016.
\bibitem{Natalia2015Al}
T. Natalia, Algebraic Generalizations of Bent Functions, In: Bent Functions, pp. 133-149, Academic Press, 2015.
\bibitem{OP2020Du}
F. \"{O}zbudak and R. M. Pelen, Duals of non-weakly regular bent functions are not weakly regular and generalization to plateaued functions, \emph{Finite Fields Appl.} vol. 64: 101668, 2020.
\bibitem{OP2022Tw}
F. \"{O}zbudak and R. M. Pelen, Two or three weight linear codes from non-weakly regular bent functions, \emph{IEEE Trans. Inf. Theory} vol. 68, no. 5, pp. 3014-3027, 2022.
\bibitem{Rothaus1976On}
O. S. Rothaus, On ``bent" functions, \emph{J. Comb. Theory Ser. A}, vol. 20, no. 3, pp. 300-305, 1976.
\bibitem{Tan2010St}
Y. Tan, A. Pott and T. Feng, Strongly regular graphs associated with ternary bent functions, \emph{J. Combin. Theory Ser. A } vol. 117, pp. 668-682, 2010.
\bibitem{VDM2010So}
E. R. Van Dam and M. Muzychuk, Some implications on amorphic association schemes, \emph{J. Combin. Theory Ser. A} vol. 117, no. 2, pp. 111-127, 2010.
\bibitem{VanDam2003St}
E. R. Van Dam, Strongly regular decompositions of the complete graph, \emph{J. Algebr. Combinatorics} vol. 17, no. 2, pp. 181-201, 2003.
\bibitem{WF2023Ne}
J. Wang and F.-W. Fu, New results on vectorial dual-bent functions and partial difference sets, \emph{Des. Codes Cryptogr.} vol. 91, no. 1, pp. 127-149, 2023.
\bibitem{WF2022On}
J. Wang and F.-W. Fu, On the duals of generalized bent functions, \emph{IEEE Trans. Inf. Theory} vol. 68, no. 7, pp. 4770-4781, 2022.
\bibitem{WFW2023Be}
J. Wang, F.-W. Fu and Y. Wei, Bent partitions, vectorial dual-bent functions and partial difference sets, \emph{IEEE Trans. Inf. Theory}, doi: 10.1109/TIT.2023.3295099.
\bibitem{WSWF2023Co}
J. Wang, Z. Shi, Y. Wei and F.-W. Fu, Constructions of linear codes with two or three weights from vectorial dual-bent functions, \emph{Discret. Math.} vol. 346, no. 8: 113448, 2023.




\end{thebibliography}
\end{document}